\documentclass[]{llncs}
\usepackage{amssymb,nicefrac}
\usepackage{amsfonts}
\usepackage{mathtools}
\usepackage{calc}
\usepackage{multicol}
\usepackage{hyperref}
\usepackage{graphicx}
\usepackage{needspace}

\usepackage{amsmath}

\usepackage{stmaryrd}
\usepackage{xspace}

\usepackage{geometry}
\usepackage{tikz}
\usepackage{pgfplots}
\usetikzlibrary{decorations.markings}
\usetikzlibrary{matrix,backgrounds,folding}
\usetikzlibrary{shapes.geometric}
\pgfdeclarelayer{edgelayer}
\pgfdeclarelayer{nodelayer}
\pgfsetlayers{background,edgelayer,nodelayer,main}
\tikzstyle{every picture}=[baseline=-0.25em]
\tikzstyle{none}=[inner sep=0mm]
\tikzstyle{zxnode}=[shape=circle, minimum width=.25cm, inner sep=0.5pt, font=\footnotesize, draw=black]
\tikzstyle{gn}=[zxnode ,fill=green]
\tikzstyle{rn}=[zxnode ,fill=red]
\tikzstyle{H box}=[rectangle,fill=yellow,draw=black,xscale=1,yscale=1,font=\small,inner sep=0.75pt,minimum width=0.15cm,minimum height=0.15cm]
\tikzstyle{ug}=[regular polygon, regular polygon sides=3, fill=red,draw=black,inner sep = 0pt,minimum width=1em]
\tikzstyle{black dot}=[inner sep=0.7mm,minimum width=0pt,minimum height=0pt,fill=black,draw=black,shape=circle]
\tikzstyle{dot}=[black dot]
\tikzstyle{white dot}=[dot,fill=white]
\tikzstyle{zwcross}=[diamond, draw, fill=gray, minimum width=0em, inner sep=1.5pt]
\tikzstyle{arrow}=[decoration={markings,mark=at position 1 with
    {\arrow[scale=1.2,>=stealth]{>}}},postaction={decorate}]

\tikzstyle{st}=[star,star points = 5, fill=white,draw=black,inner sep = 1.2pt,line width=1.2pt]
\tikzstyle{uglabel}=[rounded corners=0.2em,fill=green!20,inner sep=0.1em,font=\scriptsize, anchor=west, xshift=-0.2em, yshift=0,opacity=1]

\pgfdeclarelayer{edgelayer}
\pgfdeclarelayer{nodelayer}
\pgfsetlayers{background,edgelayer,nodelayer,main}
\tikzstyle{none}=[inner sep=0mm]
\tikzstyle{every loop}=[]

\newcommand{\tikzfig}[1]{
\input{./figures/#1.tikz}
}

\def\fig{}

\newcommand{\tikzfigc}[1]{\input{./figures/\fig/\fig_#1.tikz}}

\newcommand{\callrule}[2]{\hyperlink{r:#1}{\textnormal{(#2)}}\xspace}

\newcommand{\sncf}{$S$-\textup{CNF}\xspace}
\newcommand{\scnf}{$S$-\textup{CNF}\xspace}
\newcommand{\snf}{$S$-\textup{NF}\xspace}

\newcommand{\soo}{\callrule{rules}{S1}}
\newcommand{\stt}{\callrule{rules}{S2}}

\newcommand{\bo}{\callrule{rules}{B1}}
\newcommand{\bt}{\callrule{rules}{B2}}
\newcommand{\kt}{\callrule{rules}{K}}
\newcommand{\eu}{\callrule{rules}{EU}}
\newcommand{\h}{\callrule{rules}{H}}

\newcommand{\add}{\callrule{add}{A}}

\newcommand{\integ}{\callrule{cancellation}{cancel}}

\newcommand{\eq}[2][~~]{
#1
\underset{\substack{#2}}{=}
#1
}

\newcommand{\equi}[2][\quad]{
#1
\underset{\substack{#2}}{\iff}
#1
}
\newcommand{\fit}[1]{\resizebox{\columnwidth}{!}{#1}}
\newcommand{\interp}[1]{\left\llbracket #1 \right\rrbracket}
\newcommand{\frag}[1]{$\frac{\pi}{#1}$-frag\-ment}
\newcommand{\titlerule}[1]{\noindent
\begin{center}
\rule{(\textwidth-\widthof{#1})/2}{0.5pt}#1\rule{(\textwidth-\widthof{#1})/2}{0.5pt}
\end{center}}

\newcommand{\bra}[1]{\ensuremath{\left\langle #1 \right|}}
\newcommand{\ket}[1]{\ensuremath{\left|  #1 \right\rangle}}
\newcommand{\ketbra}[2]{\ket{#1}\!\!\bra{#2}}
\def \zx {\textnormal{ZX}\xspace}
\def \zxc {\textnormal{ZX$^*_{G}$}\xspace}

\newcommand{\gna}{\begin{tikzpicture}
	\begin{pgfonlayer}{nodelayer}
		\node [style=gn] (0) at (0, -0) {$\alpha$};
	\end{pgfonlayer}
\end{tikzpicture}}

\allowdisplaybreaks

\title{A Generic Normal Form for ZX-Diagrams and \\Application to the Rational Angle Completeness}

\author{
Emmanuel Jeandel
\and Simon Perdrix
\and Renaud Vilmart
\institute{Universit\'e de Lorraine, CNRS, Inria, LORIA, F 54000 Nancy, France}
\\\email{emmanuel.jeandel@loria.fr}$\quad$
\email{simon.perdrix@loria.fr}$\quad$
\email{renaud.vilmart@loria.fr}
}
\begin{document}

\maketitle

\begin{abstract}
Recent completeness results on the ZX-Calculus used a third-party language, namely the ZW-Calculus. As a consequence, these proofs are elegant, but sadly non-cons\-tructive. We address this issue in the following. To do so, we first describe a generic normal form for ZX-diagrams in any fragment that contains  Clifford+T quantum mechanics. We give sufficient conditions for an axiomatisation to be complete, and an algorithm to reach the normal form. Finally, we apply these results to the Clifford+T fragment and the general ZX-Calculus -- for which we already know the completeness--, but also for any fragment of rational angles: we show that the axiomatisation for Clifford+T is also complete for any fragment of dyadic angles, and that a simple new rule (called \emph{cancellation}) is necessary and sufficient  otherwise. 
\end{abstract}

\section{Introduction}

The ZX-Calculus is a powerful graphical calculus devoted to quantum information processing, introduced in 2008 by Coecke and Duncan \cite{interacting}. The language relies on two fundamental structures in quantum mechanics: the interacting observables and the phase group.
Thanks to its flexibility, the language has already been used in several topics such has the foundations of quantum mechanics \cite{toy-model-graph,duncan2016hopf}, measurement-based quantum computing \cite{mbqc,horsman2011quantum,duncan2013mbqc}, quantum error correction \cite{verifying-color-code,duncan2014verifying,chancellor2016coherent,de2017zx} ...

Quantum processes are described in the language as diagrams, providing a compact and down-to-earth visualisation. Diagrams can be manipulated through the interactive proof assistant Quantomatic \cite{quanto,kissinger2015quantomatic}. As the quantum circuits, the diagrams are \emph{universal}: whatever the considered quantum operation, there exists a ZX-diagram that describes it. This representation is however not unique: two distinct ZX-diagrams may represent the same evolution. As a consequence the language is equipped with a set of equations. These equations preserve the represented evolution: they are sound. 
 The converse of soundness is \emph{completeness}, and is much harder to get. It is achieved when, whenever two diagrams represent the same evolution, they can be transformed into each other using solely the transformation rules. 

The question of the completeness of the ZX-calculus gave rise to a series of results on various fragments of the language. A fragment corresponds to a restriction on the phase group structure: the $\frac \pi n$-fragment is made of the diagrams involving angles in $\frac \pi n\mathbb Z$ only. The $\frac \pi2$- and the $\pi$-fragments -- two non universal fragments of the language -- 
have been proved complete  \cite{pi_2-complete}; and, more recently, a complete axiomatisation has been provided for the $\frac \pi 4$-fragment  \cite{JPV}, providing  the first completeness result for an (approximately) universal fragment since this fragment corresponds to the so-called Clifford+T quantum mechanics. 
This  has been then  extended to a complete axiomatisation of the general ZX-calculus \cite{JPV-universal,NgWang}.

All these recent completeness results for (approximately) universal ZX-calculi used different versions of another graphical language called ZW-Calculus  \cite{ghz-w,zw,Amar}. 
The language describes the interactions between the only two non-equivalent kinds of entanglement between three qubits, precisely the GHZ and W states \cite{SLOCC}. In its last version \cite{Amar}, the ZW-calculus is crucially parametrised by a ring, and as a consequence admits  a natural representation of matrices over this ring: the ZW-diagram represent the structure of the matrix where some of the generators are parametrised by the entries of the matrix. This representation of matrices led to a notion of ZW-diagrams in normal forms on which the proof of completeness is built.

The ZX-calculus is instead parametrised by elements of a group, the so-called phase group structure.
As a consequence the representation of matrices (over a ring) is more involved. We introduce in this paper the first normal forms for (approximately) universal fragments of the ZX-calculus. 
This normal form is generic, depending on the considered fragment of the language. We reprove the two completeness theorems of the ZX-calculus, namely for the $\frac \pi4$-fragment and the general calculus,   but this time constructively, using the normal form in the ZX-Calculus, and hence without using a third-party language. Moreover, we prove the completeness for any fragment of rational angles: we show that for any fragment of dyadic angles (which allows for instance the exact representation of the Quantum Fourier Transform \cite{nielsen_chuang_2010}) is complete; we also show that for any other fragment of rational angles, the following new and simple rule, called \emph{cancellation}, is necessary and sufficient for completeness:
\[ \forall\alpha\neq \pi\bmod2\pi,~~\zx\vdash D_1\otimes\gna=D_2\otimes\gna\underset{\textnormal{(Cancel)}}{\implies}\zx\vdash D_1=D_2\]

\noindent \textbf{Related works.} Two completeness results on diagrammatic languages have been established recently \cite{ZH,2-qubits-zx}, independently of the present work. In  \cite{ZH}, a new language, the ZH-calculus is introduced. The ZH-calculus is intuitively an angle-free ZX-calculus augmented with a generalisation of the  H-generator  with an arbitrary number of inputs/outputs and parametrised by a complex number.  This language allows very nice and simple representation of some useful controlled operations. The authors give a completeness result based on normal forms. Like in the ZW-calculus, the entries of a complex matrix can be directly represented in a ZH-diagram while the representation of the scalars is the cornerstone -- and the main technicality -- of the normal forms in ZX-diagrams. 
In  \cite{2-qubits-zx},  the authors show that a simpler axiomatisation of the ZX-calculus is enough to prove the equivalence  of 2-qubit Clifford+T circuits. Surprisingly, the proposed axiomatisation is based on the use of diagrams which are not in the $\frac \pi4$-fragment whereas all 2-qubit Clifford+T circuits are in this fragment. 

\noindent\textbf{Structure of the paper.} We first present the ZX-Calculus in Section \ref{sec:zx-calculus}. We then give the general structure of the normal form in Section \ref{sec:controlled-states}, and sufficient conditions for obtaining the completeness. We apply this for the general ZX-Calculus in Section \ref{sec:general-zx}, for rational angles 
 in Section \ref{sec:pi_4n-fragments} and in the particular case of the dyadic angles 
  in Section \ref{sec:pi_2^n-fragments}.

\section{The ZX-Calculus}
\label{sec:zx-calculus}

\subsection{Diagrams and Standard Interpretation}
A ZX-diagram $D:k\to l$ with $k$ inputs and $l$ outputs is generated by:
\begin{center}
\bgroup
\def\arraystretch{3.5}
{\begin{tabular}{|@{~}cc|@{~}cc|@{~}cc|@{~}cc|}
\hline
$R_Z^{(n,m)}(\alpha):n\to m$ & 
\input{./figures/gn-alpha.tikz}
 & $H:1\to 1$ & 
\begin{tikzpicture}
	\begin{pgfonlayer}{nodelayer}
		\node [style={H box}] (0) at (0, 0) {};
		\node [style=none] (1) at (0, 0.5) {};
		\node [style=none] (2) at (0, -0.5) {};
	\end{pgfonlayer}
	\begin{pgfonlayer}{edgelayer}
		\draw (2.center) to (1.center);
	\end{pgfonlayer}
\end{tikzpicture}
 & $\mathbb{I}:1\to 1$ & 
\begin{tikzpicture}
	\begin{pgfonlayer}{nodelayer}
		\node [style=none] (0) at (0, 0.2499999) {};
		\node [style=none] (1) at (0, -0.2499999) {};
	\end{pgfonlayer}
	\begin{pgfonlayer}{edgelayer}
		\draw (0.center) to (1.center);
	\end{pgfonlayer}
\end{tikzpicture}
 & $\epsilon:2\to 0$ & 
\begin{tikzpicture}
	\begin{pgfonlayer}{nodelayer}
		\node [style=none] (0) at (-0.2500001, 0.2500001) {};
		\node [style=none] (1) at (0.2500001, 0.2500001) {};
	\end{pgfonlayer}
	\begin{pgfonlayer}{edgelayer}
		\draw [bend right=90, looseness=1.75] (0.center) to (1.center);
	\end{pgfonlayer}
\end{tikzpicture}
\\
\hline
$R_X^{(n,m)}(\alpha):n\to m$ & 
\input{./figures/rn-alpha.tikz}
 & $e:0\to 0$ & 
\input{./figures/empty-diagram.tikz}
& $\sigma:2\to 2$ & 
\input{./figures/crossing.tikz}
& $\eta:0\to 2$ & 
\begin{tikzpicture}
	\begin{pgfonlayer}{nodelayer}
		\node [style=none] (a0) at (-0.2500001, -0) {};
		\node [style=none] (a1) at (0.2500001, -0) {};
		\node [style=none] (a2) at (0, 0.25) {};
	\end{pgfonlayer}
	\begin{pgfonlayer}{edgelayer}
		\draw [bend left=90, looseness=1.75] (a0.center) to (a1.center);
	\end{pgfonlayer}
\end{tikzpicture}
\\\hline
\end{tabular}}
\egroup\\
where $n,m\in \mathbb{N}$, $\alpha \in \mathbb{R}$, and the generator $e$ is the empty diagram.
\end{center}
and the two compositions:
\begin{itemize}
\item Spacial Composition: for any $D_1:a\to b$ and $D_2:c\to d$, $D_1\otimes D_2:a+c\to b+d$ consists in placing $D_1$ and $D_2$ side by side, $D_2$ on the right of $D_1$.
\item Sequential Composition: for any $D_1:a\to b$ and $D_2:b\to c$, $D_2\circ D_1:a\to c$ consists in placing $D_1$ on the top of $D_2$, connecting the outputs of $D_1$ to the inputs of $D_2$.
\end{itemize}

The ZX-Calculus comes with a way of interpreting its diagrams as matrices: The standard interpretation of the ZX-diagrams associates to any diagram $D:n\to m$ a linear map $\interp{D}:\mathbb{C}^{2^n}\to\mathbb{C}^{2^m}$ inductively defined as follows:\\
\begin{minipage}{\columnwidth}
\titlerule{$\interp{.}$}
$$ \interp{D_1\otimes D_2}:=\interp{D_1}\otimes\interp{D_2} \qquad 
\interp{D_2\circ D_1}:=\interp{D_2}\circ\interp{D_1}$$
\end{minipage}
$$\interp{
\input{./figures/empty-diagram.tikz}
~}:=1 \qquad
\interp{~
\begin{tikzpicture}
	\begin{pgfonlayer}{nodelayer}
		\node [style=none] (0) at (0, 0.2499999) {};
		\node [style=none] (1) at (0, -0.2499999) {};
	\end{pgfonlayer}
	\begin{pgfonlayer}{edgelayer}
		\draw (0.center) to (1.center);
	\end{pgfonlayer}
\end{tikzpicture}
~~}:= \ketbra{0}{0}+\ketbra{1}{1}\qquad
\interp{~
\begin{tikzpicture}
	\begin{pgfonlayer}{nodelayer}
		\node [style={H box}] (0) at (0, 0) {};
		\node [style=none] (1) at (0, 0.5) {};
		\node [style=none] (2) at (0, -0.5) {};
	\end{pgfonlayer}
	\begin{pgfonlayer}{edgelayer}
		\draw (2.center) to (1.center);
	\end{pgfonlayer}
\end{tikzpicture}
~}:= \ketbra{+}{0}+\ketbra{-}{1}\qquad
\interp{
\input{./figures/crossing.tikz}
}:= \hspace*{-0.7em}\sum_{i,j\in\{0,1\}} \hspace*{-0.7em}\ketbra{ij}{ji} $$
$$\interp{\raisebox{-0.35em}{$
\begin{tikzpicture}
	\begin{pgfonlayer}{nodelayer}
		\node [style=none] (a0) at (-0.2500001, -0) {};
		\node [style=none] (a1) at (0.2500001, -0) {};
		\node [style=none] (a2) at (0, 0.25) {};
	\end{pgfonlayer}
	\begin{pgfonlayer}{edgelayer}
		\draw [bend left=90, looseness=1.75] (a0.center) to (a1.center);
	\end{pgfonlayer}
\end{tikzpicture}
$}}:= \ket{00}+\ket{11}\qquad
\interp{\raisebox{-0.25em}{$
\begin{tikzpicture}
	\begin{pgfonlayer}{nodelayer}
		\node [style=none] (0) at (-0.2500001, 0.2500001) {};
		\node [style=none] (1) at (0.2500001, 0.2500001) {};
	\end{pgfonlayer}
	\begin{pgfonlayer}{edgelayer}
		\draw [bend right=90, looseness=1.75] (0.center) to (1.center);
	\end{pgfonlayer}
\end{tikzpicture}
$}}:= \bra{00}+\bra{11}\qquad
\interp{\begin{tikzpicture}
	\begin{pgfonlayer}{nodelayer}
		\node [style=gn] (0) at (0, -0) {$\alpha$};
	\end{pgfonlayer}
\end{tikzpicture}}=
\interp{\begin{tikzpicture}
	\begin{pgfonlayer}{nodelayer}
		\node [style=rn] (0) at (0, -0) {$\alpha$};
	\end{pgfonlayer}
\end{tikzpicture}}:=1+e^{i\alpha} $$
For any $n,m$ such that $n+m>0$:\\
\begin{minipage}{\columnwidth}

$$\interp{
\input{./figures/gn-alpha.tikz}
}:= \ketbra{0^m}{0^n}+e^{i\alpha}\ketbra{1^m}{1^n} \qquad
\interp{
\input{./figures/rn-alpha.tikz}
}:= \ketbra{+^m}{+^n}+e^{i\alpha}\ketbra{-^m}{-^n}
$$
\rule{\columnwidth}{0.5pt}
\end{minipage}\\
where $\ket{+}:=\frac{\ket{0}+\ket{1}}{\sqrt{2}}$, $\ket{-}:=\frac{\ket{0}-\ket{1}}{\sqrt{2}}$ and $\ket{i^n}:=|\overbrace{i\cdots i}^n\rangle$.

To simplify, the red and green nodes will be represented empty when holding a 0 angle:
\[ \scalebox{0.9}{
\input{./figures/gn-empty-is-gn-zero.tikz}
} \qquad\text{and}\qquad \scalebox{0.9}{
\input{./figures/rn-empty-is-rn-zero.tikz}
} \]
ZX-Diagrams are universal:
\[\forall A\in \mathbb{C}^{2^n}\times\mathbb{C}^{2^m},~~\exists D:n\to m,~~ \interp{D}=A\]
This is true for \emph{general} ZX-diagrams i.e.~where angles are in $\mathbb{R}/2\pi\mathbb{Z}$. However, it is convenient to consider restrictions of the language -- called fragments -- that are finitely generated. Let $G$ be an additive subgroup of $\mathbb{R}/2\pi\mathbb{Z}$. It is easy to see that the standard interpretation $\interp{.}$ maps diagrams of the fragment $G$ to matrices over $\mathcal{R}_G:=\mathbb{Z}\left[\frac{1}{\sqrt{2}},e^{iG}\right]$, that is, the smallest subring of $\mathbb{C}$ that contains the integers $\mathbb{Z}$, $\frac{1}{\sqrt{2}}$, and the set $\{e^{i\alpha}~|~\alpha\in G\}$.

However, in general, all matrices in $\mathcal{R}_G$ are not expressible with a diagram of the fragment $G$. For instance, $\pi\mathbb{Z}_2$ and $\frac{\pi}{2}\mathbb{Z}_4$ are not universal \cite{clifford-not-universal}. However, we will show in the following that if $\frac{\pi}{4}\in G$, then the fragment $G$ is universal for matrices over $\mathcal{R}_G$.

\begin{definition} Let $\mathcal G$ be the set of all additive subgroup $G$ of $\mathbb{R}/2\pi\mathbb{Z}$ such that $\frac \pi 4\in G$. 
\end{definition}

\subsection{Calculus}

Since the diagrammatic representation of a matrix is not unique with ZX-diagrams, the calculus comes with a set of axioms that can be used to rewrite diagrams as equivalent ones (diagrams that represent the same matrix). The axioms for the \frag4 of the calculus are represented in Figure \ref{fig:ZX_rules}.

\begin{figure*}[!htb]
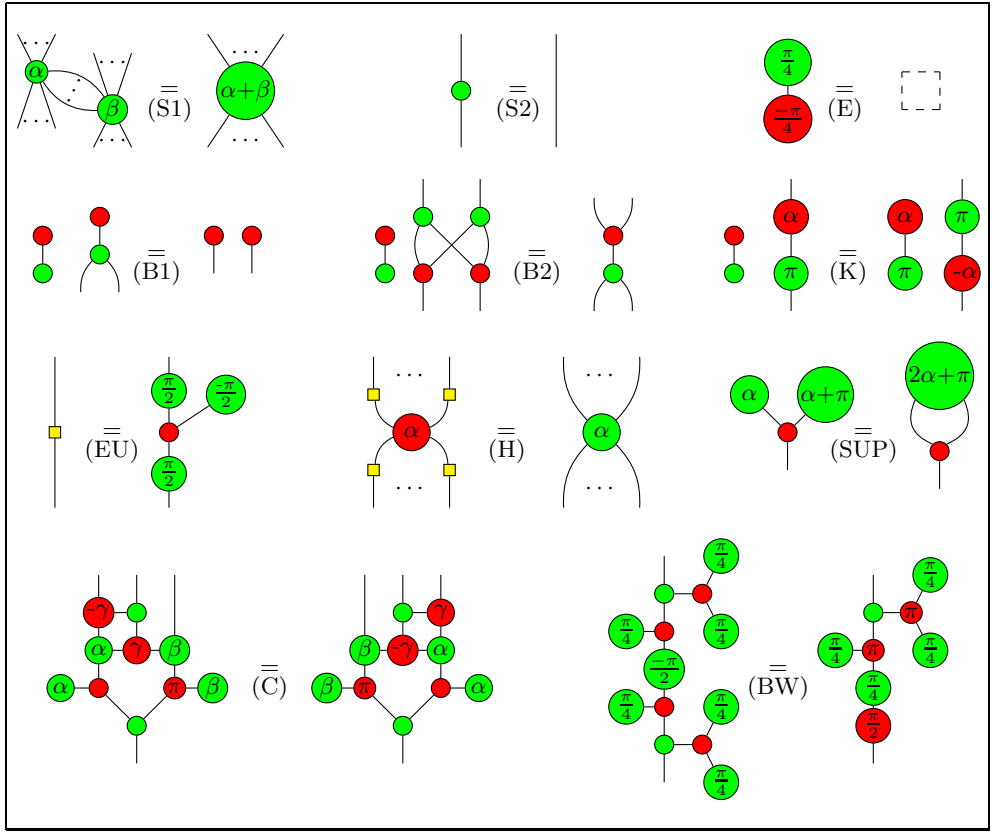

 \centering
 \hypertarget{r:rules}{}
 \begin{tabular}{|c@{$\qquad\quad~$}c@{$\qquad\quad~$}c@{~}|}
   \hline
   && \\
   
\input{./figures/spider-1.tikz}
&
\input{./figures/s2-simple.tikz}
&
\input{./figures/bicolor_pi_4_eq_empty.tikz}
\\
   && \\
   
\input{./figures/b1s.tikz}
&
\input{./figures/b2s.tikz}
&
\input{./figures/k2s.tikz}
\\
   && \\
   
\input{./figures/euler-decomp-scalar-free.tikz}
&
\input{./figures/h2.tikz}
&
\input{./figures/former-supp.tikz}
\\
   && \\
   \multicolumn{3}{|c|}{
\input{./figures/commutation-of-controls-general-simplified.tikz}
$\qquad\qquad$
\input{./figures/BW-simplified.tikz}
}\\
   && \\
   \hline
  \end{tabular}
 \caption[]{Set of rules \zx for the ZX-Calculus with scalars. All of these rules also hold when flipped upside-down, or with the colours red and green swapped. The right-hand side of (E) is an empty diagram. (...) denote zero or more wires, while (\protect\rotatebox{45}{\raisebox{-0.4em}{$\cdots$}}) denote one or more wires. $\zx_G$ is obtained when constraining the angles $\alpha,\beta,\gamma\in G$.
  }
 \label{fig:ZX_rules}
\end{figure*}

To these axioms are added a set of transformation rules aggregated under the paradigm \emph{Only Topology Matters}. It means that the wires can be bent at will, and that the in\-puts/out\-puts of the generators $R_Z$, $R_X$ and $H$ can be reordered at will. What matters is solely the connectivity between two nodes. Such axioms are:
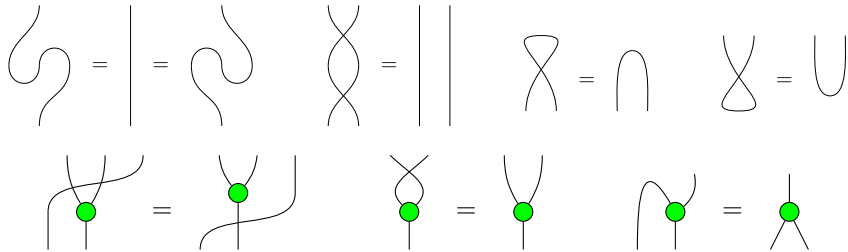
\[\scalebox{0.8}{
\input{./figures/bent-wire.tikz}
}\]
\[
\input{./figures/bent-wire-2.tikz}
\]

When one can transform one diagram $D_1$ into another $D_2$ using only the rules of the ZX-Calculus, we write $\zx\vdash D_1=D_2$, which can be done by applying axioms locally. Indeed, for any diagrams $D$, $D_1$ and $D_2$, if $\zx\vdash D_1=D_2$ then:
\begin{itemize}
\item $\zx\vdash D_1\circ D = D_2\circ D$ 
\item $\zx\vdash D\circ D_1 = D\circ D_2$
\item $\zx\vdash D_1\otimes D = D_2\otimes D$ 
\item $\zx\vdash D\otimes D_1 = D\otimes D_2$
\end{itemize}

The local application of axioms is \emph{sound}: it preserves the represented matrix. The converse of soundness is \emph{completeness}. The language is complete if we can transform two diagrams into one another as long as they represent the same matrix. In other words, the language is complete if it captures all the power of quantum mechanics.

We call \emph{\frag{q}} the restriction of the language where angle $\alpha$ can only be a multiple of $\frac{\pi}{q}$ in $R_Z$ and $R_X$, and we write the resulting set of axioms with subscript $(.)_{\frac{\pi}{q}}$. More generally, if $G$ is an additive subgroup of $\mathbb{R}/2\pi\mathbb{Z}$, then we denote the resulting set of axioms $(.)_G$. By convention, when there is no subscript, we refer to the general ZX-Calculus, e.g. $\zx:=\zx_{\mathbb{R}}$.

The set of axioms given in Figure \ref{fig:ZX_rules} is known to be complete for the \frag4, the first approximately universal fragment of the ZX-Calculus \cite{JPV}. It is also known that one only has to add the axiom \add (Figure \ref{fig:add-axiom}) to make the ZX-Calculus complete in general \cite{JPV-universal}. When considering a set of rules augmented with an additional axiom, we use the superscript notation. For instance, the complete set of rules for the general ZX-Calculus is denoted $\zx^{\textnormal{A}}$.

\begin{figure}[htb]
\noindent\begin{minipage}{0.45\columnwidth}
 \hypertarget{r:add}{}
\begin{center}
\boxed{
\input{./figures/add-axiom.tikz}
}
\end{center}
\caption[]{Additional axiom for the completeness of the ZX-Calculus in general.}\label{fig:add-axiom}
\end{minipage}\hfill
\begin{minipage}{0.45\columnwidth}
\centering

\input{./figures/triangle-decomp.tikz}

\caption[]{Triangle node}\label{fig:triangle}
\end{minipage}
\end{figure}

Introduced in \cite{JPV} as a syntactic sugar and used as a generator in \cite{NgWang,NgWang-clifford+t,zx-toffoli} is the so-called \emph{triangle}: \scalebox{0.8}{
\begin{tikzpicture}
	\begin{pgfonlayer}{nodelayer}
		\node [style=ug] (0) at (0, -0) {};
		\node [style=none] (1) at (0, 0.5000001) {};
		\node [style=none] (2) at (0, -0.5000001) {};
		\node [style=none] (3) at (0, -0.7499999) {};
		\node [style=none] (4) at (0, 0.7500001) {};
	\end{pgfonlayer}
	\begin{pgfonlayer}{edgelayer}
		\draw (1.center) to (2.center);
	\end{pgfonlayer}
\end{tikzpicture}
}.
It stands for a ZX-diagram of the \frag4 (Figure \ref{fig:triangle}),
and is used in numerous lemmas as it represents a non-trivial quantum process with integer coefficients: $\ketbra{0}{0}+\ketbra{1}{1}+\ketbra{0}{1}$.

\section{Controlled States and Normal Form}
\label{sec:controlled-states}

\subsection{The Transistor and its Algebra}

We first define another syntactic sugar, which will be used in the normal form:

\begin{definition}
We define the \emph{transistor} as the three legged diagram:
$$
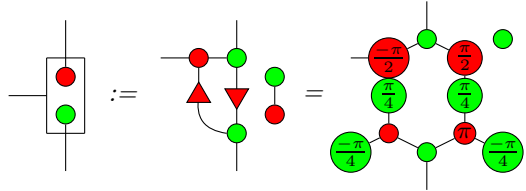
\input{./figures/control-separation.tikz}
$$
\end{definition}
Thanks to Lemmas \ref{lem:2-triangle-cycle} and \ref{lem:2-cycle-triangle-with-one-not}, one can check that:
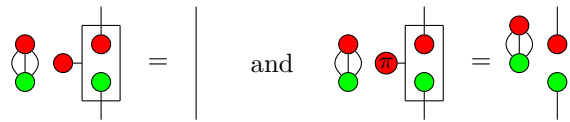
\[
\input{./figures/control-sep-0-1.tikz}
\]
It can be seen as a control of a switch: if $\ket{0}$ is plugged on the left, the right wire is intact, but if $\ket{1}$ is plugged on the left, the right wire is ``opened'' by the operation. A classical transistor \reflectbox{\scalebox{-0.6}{
\input{./figures/classical-transistor.tikz}
}} has a similar mechanics: if some electrical current is applied on the control side, it allows for current to pass through the vertical wire, otherwise it acts as an open switch. If we want to assimilate $\ket{0}$ to ``no current'' and $\ket{1}$ to ``current'', we actually have to apply NOT on the control wire.

\begin{proposition}
\label{prop:monoid}
$\left(
\input{./figures/anti-control-sep-upside-down.tikz}
~~,~~
\input{./figures/ket-1.tikz}
\right)$ forms a commutative monoid:
$$ \zx_{\frac{\pi}{4}}\vdash
\left(\scalebox{0.8}{
\input{./figures/AND-monoid-commutative.tikz}
}\right),
\left(\scalebox{0.8}{
\input{./figures/AND-monoid-unit.tikz}
}\right),
\left(\scalebox{0.8}{
\input{./figures/AND-monoid-associative.tikz}
}\right) $$
\end{proposition}

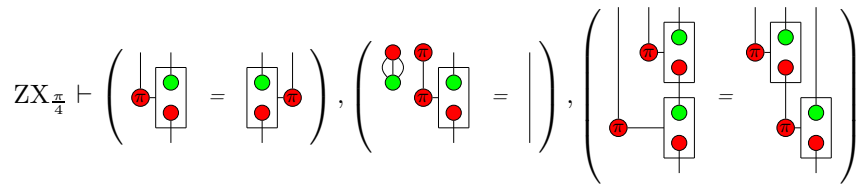
\begin{proposition}
\label{prop:bialgebra}
$\left(
\input{./figures/anti-control-sep.tikz}
~~,~~
\input{./figures/bra-1.tikz}
\right)$ and $\left(
\input{./figures/gn-2-1.tikz}
~~,~~
\begin{tikzpicture}
	\begin{pgfonlayer}{nodelayer}
		\node [style=gn] (0) at (0, 0.25) {};
		\node [style=none] (1) at (0, -0.25) {};
	\end{pgfonlayer}
	\begin{pgfonlayer}{edgelayer}
		\draw [style=none] (0) to (1.center);
	\end{pgfonlayer}
\end{tikzpicture}
\right)$ form a bialgebra:
$$ \zx_{\frac{\pi}{4}}\vdash
\left(\scalebox{0.8}{
\input{./figures/bialgebra-unit-counit.tikz}
}\right),
\left(\scalebox{0.8}{
\input{./figures/bialgebra-multiplication-counit.tikz}
}\right),
\left(\scalebox{0.8}{
\input{./figures/bialgebra-comultiplication-unit.tikz}
}\right),
\left(\scalebox{0.8}{
\input{./figures/bialgebra-multiplication-comultiplication.tikz}
}\right)
$$
\end{proposition}

\begin{remark}
\scalebox{0.8}{
\input{./figures/anti-control-sep-upside-down.tikz}
} can be seen as an AND gate (notice that when plugging \scalebox{0.8}{
\input{./figures/k-pi-l-pi.tikz}
}, the result is \scalebox{0.8}{
\input{./figures/kl-pi.tikz}
}, when $k,\ell\in\{0,1\}$). As such, it has been used in \cite{JPV,zx-toffoli} to create the Toffoli gate. The previous two propositions where observed as tensor network transformations with AND gates in \cite{tensor-network}.
\end{remark}

\subsection{Controlled states}

In this section, we present the cornerstone of the normal forms: the controlled states. Controlled states form a particular family of \zx diagrams with a single input and $n$  outputs, their interpretation should map $\ket 0$ to the uniform superposition $ \sum_{x\in \{0,1\}^n}\ket x$. Intuitively, a controlled state $D:1\to n$ is just an encoding for the state $\interp D \ket 1$.

\begin{definition}[Controlled states]
A \zx-diagram $D:1\to n$ is a \emph{controlled state} if  $\interp D \ket 0 = \sum_{x\in \{0,1\}^n}\ket x$. 
\end{definition}

A controlled state with no output is called a controlled scalar:

\begin{definition}[Controlled scalars]
A \zx-diagram $D:1\to 0$ is a \emph{controlled scalar} if $\interp D\ket 0 = 1$. 
\end{definition}

\noindent For instance \!\!\!
\input{./figures/undemi.tikz}
 is a controlled scalar encoding $\frac 12$: $\interp{
\input{./figures/undemi.tikz}
}\ket x = \begin{cases}1&\text{if $x=0$}\\\frac12&\text{if $x=1$}\end{cases}$.

We introduce other examples of controlled scalars, parameterised by integer polynomials:

\begin{definition}
For any $G\in \mathcal G$ and any $\alpha \in G$, let $\mathrm{\Gamma}_\alpha:\mathbb{Z}[X]\to \zx_G$ be the map which associates to any polynomial $P$ a \zx-diagram $\mathrm{\Gamma}_\alpha (P):1\to 0$, 
  inductively defined as $0\mapsto  \rotatebox[origin=c]{180}{
\input{./figures/ket-0.tikz}
}$,  and $\forall  a\in \mathbb N\setminus \{0\}, \forall b\in \{0,1\}, \forall k\in \mathbb N$, and $\forall P\in\mathbb{Z}[X]$ such that $\deg(P)<k$, 
\[\begin{array}{r@{~~\mapsto}l}
{(-1)^baX^k+ P} & ~~
\input{./figures/delta_alpha-def.tikz}

\end{array}\qquad
\left(\text{where }
\input{./figures/power-composition.tikz}
\right)\]
\end{definition}

For any integer polynomial $P$, the corresponding diagram $\mathrm{\Gamma}_\alpha(P)$ is  a controlled scalar encoding the scalar $P(e^{i\alpha})$:
\begin{lemma} $\forall G\in \mathcal G$,  $\forall \alpha \in G$, and $\forall P\in \mathbb Z[X]$,  $\interp{\mathrm{\Gamma}_\alpha(P)}\ket x =\begin{cases}1&\text{if $x=0$}\\ P(e^{i\alpha})&\text{if $x=1$}\end{cases}$.\\
\end{lemma}

Whereas it is not obvious in the \zx-calculus to add two given diagrams, a fundamental property of controlled states is that they can   be freely added and multiplied (according to the  entrywise product a.k.a. the Hadamard product or Schur product) as follows:

\begin{lemma}[Sum and Product]\label{lem:sumprod}
For any controlled states $D_0,D_1:1\to n$, \[D_{\textup{sum}}:=
\input{./figures/sum.tikz}
\qquad \qquad D_{\textup{prod}}:=
\input{./figures/prod.tikz}
\] are controlled states  such that $\interp{D_{\textup{sum}}}\ket 1 = \interp{D_0}\ket 1+\interp{D_1}\ket 1$ and $\interp{D_{\textup{prod}}}\ket 1 = (\interp{D_0}\ket 1) \bullet (\interp{D_1}\ket 1)$, where $.\bullet.$ is the entrywise product. 
\end{lemma}

\subsection{Normal forms}

Among the family of controlled state diagrams,  we define those which are in normal form. Our definition of normal form is \emph{generic} in the sense that it is defined with respect to a given set of controlled scalars. Intuitively the choice of these controlled scalars  depends on the considered fragment of the language, as detailed in the next sections. 

\begin{definition}[Controlled Normal  Form]
Given a set $S$ of controlled scalars, the diagrams in \emph{normal controlled form} with respect to $S$ (\sncf) are inductively defined as follows: 
\begin{itemize}
\item $\forall D\in S$, $D$ is in \sncf;  
\item $\forall D_0, D_1$ in \sncf, $
\input{./figures/CNF-inductive.tikz}
$ is in \sncf.
\end{itemize}
A diagram $D$ in \sncf~is depicted ${
\input{./figures/CNF.tikz}
}$. 
\end{definition}

One can double check that diagrams in controlled normal form are actually controlled states: if $D:1\to n$ is in \scnf, $\interp D \ket 0 = \sum_{x\in \{0,1\}^n }\ket x$ (Lemma \ref{lem:defCNF} in appendix). 

We are now ready to give a definition of diagrams in normal form, based on the diagrams in controlled normal forms:

\begin{definition}[Normal Form]
Given a set $S$ of controlled scalars, for any $n,m \in \mathbb N$, and any $D:1\to n+m$ in \sncf, $
\input{./figures/SNF.tikz}
$
is in \emph{normal form} with respect to $S$ (\snf). 
\end{definition}

\subsection{Universality}

While the main application of the notion of normal form is to prove completeness results (in the next sections), our first application  is to prove the universality of $\zx_G$ for any $G\in \mathcal G$. First notice that the universality of $\zx_G$ can be reduced to the existence of an appropriate set of controlled scalars:

\begin{lemma}[Sufficient condition for universality]\label{lem:univ}
Given $G\in \mathcal G$, if  $\exists S\subseteq ZX_G$  a set  of controlled scalars 
 such that  the map $\eta : S\to \mathcal R_G = D\mapsto \interp D \ket 1$ is surjective,  then $\zx_G$ is universal. 
\end{lemma}

\begin{theorem}\label{thm:univ}
For any  $G\in \mathcal G$, $ZX_G$ is universal: 
 $$\forall n,m\in \mathbb N, \forall M\in \mathcal R_G^{2^n\times 2^m}, \exists D\in ZX_G, \interp D = M$$
\end{theorem}

\begin{proof}
Let $S\subseteq \zx_G$ be the set of all controlled scalars.  According to Lemma \ref{lem:univ} it suffices to show that $\eta :S\to \mathcal R_G$ is onto. Let $x\in \mathcal R_G$,  there exist $p\in \mathbb N$, $\alpha_0, \ldots, \alpha_k \in G$, and $P_0\ldots P_k\in \mathbb Z[X]$ such that $x=\frac{1}{2^p}\sum_{j=0}^k P_j(e^{i\alpha_j})$. Since $\mathrm{\Gamma}_{\alpha_j}(P_j)$ encodes $P_j(e^{i\alpha_j})$, 
\input{./figures/undemi.tikz}
 encodes $\frac 12$ and they can be added and multiplied according to Lemma \ref{lem:sumprod}, there exists a diagram $D\in S$ such that $\interp D \ket 1 = x$.
\qed\end{proof}

\subsection{A sufficient condition for completeness}

The \emph{controlled states} give a generic internal structure for a diagram in normal form, by separating the coefficients of the process -- related to the considered fragment -- from the way they are combined -- which is done in the \frag4. Hence, all the sound operations on the \emph{structure} of the normal forms should be doable using the set of rules \zx by \cite{JPV}. The completeness for broader fragments is then reduced to the capacity to apply elementary operations on coefficients:

\begin{theorem}[Sufficient condition for completeness]\label{them:completeness}
Given $G\in \mathcal G$, $\zx_G$ is complete if  $\exists S\subseteq \zx_G$  a set  of controlled scalars such that $\eta : S\to \mathcal R_G= D\mapsto \interp D \ket 1$ is bijective, and the following equations hold: $\forall \alpha \in G, \forall x,y\in \mathcal R_G$,
\[
\input{./figures/condition-on-exp-i-alpha-d.tikz}
\qquad \qquad
\input{./figures/prod-of-control-scalars-d.tikz}
\qquad \qquad
\input{./figures/sum-of-control-scalars-d.tikz}
\]
\end{theorem}

Before proving Theorem \ref{them:completeness}, notice that all the above equations are involving diagrams with a single input and no output, thus 
 for any fragment the completeness reduces to the completeness for diagrams with 1 input and no output, or equivalently to diagrams representing $1$-qubit state preparations which have no input and a single output:

\begin{corollary}
For any  $G\in \mathcal G$, $\zx_G$ is complete if and only if it is complete  for $1$-qubit state preparations, i.e. for all diagrams with no input and a single output. 
\end{corollary}

Notice that thanks to the hypothesis of Theorem \ref{them:completeness}, one can associate to any state $\ket \varphi \in \mathcal R_G^{2^n}$ a diagram $\Lambda(\ket \varphi)$ in $S$-CNF, and to any evolution $M\in \mathcal R_G^{2^n\times 2^m}$, a diagram $\lambda(M)$ in $S$-NF: 

\begin{definition}
With the hypothesis of Theorem \ref{them:completeness}, let $\Lambda :  \bigcup\limits_{n\in \mathbb N} \mathcal R_G^{2^n} \to \text{\scnf}$ and $\lambda : \bigcup\limits_{n,m\in \mathbb N} \mathcal R_G^{2^n\times 2^m} \to \text{\snf}$ be defined as follows:
\begin{itemize}\def\fig{concatenation-of-states}
\item $\Lambda (x) := \eta^{-1}(x)$ if $x\in \mathcal R_G$,
\item $\Lambda(\ket 0 \otimes \ket {\psi_0}+ \ket 1 \otimes \ket {\psi_1}) := \def\fig{concatenation-of-states}\begin{tikzpicture}
	\begin{pgfonlayer}{nodelayer}
		\node [style=gn] (17)  at (-0.625, 0.5) {};
		\node [style=uglabel] (18)  at (0.625, 0.0) {$\frac{2p-1}{2}\pi$};
		\node [style=gn] (19)  at (0.375, 0.0) {};
		\node [style=uglabel] (20)  at (-0.375, 0.5) {$\frac{2p-1}{2}\pi{+}\pi$};
		\node [style=rn] (29)  at (-0.125, -0.5) {};
		\node [style=gn] (30)  at (-0.125, 0.0) {};
		\node [style=gn] (31)  at (-0.625, 0.0) {};
		\node [style=rn] (32)  at (-0.625, -0.5) {};
	\end{pgfonlayer}
	\begin{pgfonlayer}{edgelayer}
		\draw [style=none, bend right=45, looseness=1.25] (29) to (30);
		\draw [style=none] (29) to (30);
		\draw [style=none, bend right=45, looseness=1.25] (30) to (29);
		\draw [style=none, bend right=45, looseness=1.25] (31) to (32);
		\draw [style=none, bend right=45, looseness=1.25] (32) to (31);
		\draw [style=none] (32) to (31);
	\end{pgfonlayer}
\end{tikzpicture}$
\item $\lambda\left(\sum\limits_{\substack{x\in \{0,1\}^n\\y\in \{0,1\}^m}}\alpha_{x,y} \ket y \bra x\right) := 
\input{./figures/SNF-2.tikz}
$, where $D = \Lambda \left(\sum\limits_{\substack{x\in \{0,1\}^n\\y\in \{0,1\}^m}} \alpha_{x,y}\ket x \ket y\right)$
\end{itemize}
\end{definition} 

The proof of Theorem \ref{them:completeness} consists in showing that any diagram can be transformed into a diagram in $S$-normal form. The proof is inductive: every generator of the language can be set in $S$-normal form, moreover both the parallel and sequential compositions of $S$-normal forms can be transformed into diagrams in $S$-normal form.

\begin{lemma}
\label{scnf:tensor}
With the hypothesis of Theorem \ref{them:completeness}, for any $D_0, D_1$ in \snf, $D_0\otimes D_1$ can be transformed into a diagram in \snf. 
\end{lemma}

\begin{lemma}
\label{scnf:composition}
With the hypothesis of Theorem \ref{them:completeness}, for any $D_0:n\to m$ and $D_1:m\to k$ in \snf, $D_1\circ D_0:n\to k$ can be transformed into a diagram in \snf. 
\end{lemma}

\begin{lemma}
\label{scnf:generators}
With the hypothesis of Theorem \ref{them:completeness}, each generator can be transformed into a diagram in \snf.
\end{lemma}

The Lemmas are proven in appendix. 

In the next sections, we will consider several fragments of the ZX-calculus for which we will exhibit a diagrammatic representation of controlled states. For some fragments, 
 the above equations are provable, implying the completeness of the ZX-calculus for these fragments. For other fragments, we will need the help of some additional axioms to prove the above equations,  implying the completeness of a ZX-calculus augmented with these additional axioms.

\section{Normal Forms with Arbitrary Angles}
\label{sec:general-zx}

In the case of the general ZX-Calculus, we know \cite{JPV-universal} that the language is complete with the set of rules in Figure \ref{fig:ZX_rules} enriched with the axiom \add. Hence, we choose our set of rules to be precisely this set, denoted ZX$^{\textnormal A}$.

\begin{definition} Let $\Lambda_{\mathbb{R}}:\mathbb{C}\to\zx[1,0]$ be the map defined as:
\begin{itemize}
\item $\Lambda_{\mathbb{R}}(0)= \!\!\rotatebox[origin=c]{180}{
\input{./figures/ket-0.tikz}
}$
\item $\forall \rho>0,\forall \theta\in   [0,2\pi)$, $\Lambda_{\mathbb{R}}(\rho e^{i\theta}):= 
\input{./figures/control-x-base-case-2.tikz}
$
 \end{itemize}
 and $S_{\mathbb R}:=\{\Lambda_{\mathbb{R}}(x)~|~x\in\mathbb{C}\}$.
\end{definition}

\begin{lemma} For any $x\in \mathbb C$, $\Lambda_{\mathbb{R}}(x)$ is a controlled scalar, and $\interp{\Lambda_{\mathbb{R}}(x)} \ket 1= x$.
\end{lemma}

\begin{lemma}
The map $\eta_{\mathbb{R}}^{~}:S_{\mathbb{R}}\to \mathcal{R}_G = D\to \interp{D}\ket{1}$ is bijective, and $\Lambda_{\mathbb{R}}=\eta_{\mathbb{R}}^{\!\text{-}1}$. Moreover:
\[\zx^{\textnormal A}\vdash\left(
\input{./figures/condition-on-exp-i-alpha-d-r.tikz}
\right),\left(
\input{./figures/prod-of-control-scalars-d-r.tikz}
\right),\left(
\input{./figures/sum-of-control-scalars-d-r.tikz}
\right)\]
\end{lemma}

\begin{theorem}
\label{thm:general-completeness}
The general ZX-Calculus with set of rules $\zx^{\textnormal A}$ is complete, and any ZX-diagram can be put into a normal form with respect to $S_{\mathbb R}$.
\end{theorem}

\section{Completeness and Normal Forms with Rational Angles}
\label{sec:pi_4n-fragments}

In this section, we consider the case where the angles are rational multiples of $\pi$, i.e.~fragments $G\in \mathcal G_{\mathbb Q} := \{G\in \mathcal G~|~G\subseteq \mathbb Q\pi\}$. Among the rational angles, dyadic angles, i.e. $\mathcal G_{\mathbb D} := \{G\in \mathcal G~|~G\subseteq \mathbb D\pi\}$, where $\mathbb D := \{\frac p{2^q} ~|~ p,q \in \mathbb N\}$ enjoys some particular properties, and are considered in details in the next section.

\subsection{Incompleteness and a new  rule for cancelling scalars}

An interesting set of equations come from the controlled scalars parametrised by integer polynomials, more precisely from those parametrised by cyclotomic polynomials. Indeed for any $n>0$, $\interp{\mathrm{\Gamma}_{\frac {2\pi} n}(\phi_n)}\ket 1 = \phi_n(e^{\frac{i2\pi}n}) = 0$ (where $\phi_{n}$ is the $n^{\textnormal{th}}$ cyclotomic polynomial), thus $\interp{\mathrm{\Gamma}_{\frac {2\pi} n}(\Phi_n)} = \interp{
\input{./figures/bra-0.tikz}
}$. However, the corresponding equations are not provable in \zx  when $n=8p$ with $p$  an odd prime number, implying the incompleteness of any fragment of rational angles which contains at least one angle of the form $\frac \pi{4p}$:

\begin{lemma}[Incompleteness]
\label{lem:zx-not-complete}
For any $G\in \mathcal G_{\mathbb Q} \setminus \mathcal G_{\mathbb D}$, there exists an odd prime number $p$ such that $\mathrm{\Gamma}_{\frac {\pi} {4p}}(\Phi_{8p}) \in \zx_G$ and 
\[\zx_{\frac \pi{4p}}\not\vdash  \mathrm{\Gamma}_{\frac {\pi} {4p}}(\Phi_{8p}) =  
\input{./figures/bra-0.tikz}
\]
\end{lemma}

Notice that a similar proof of incompleteness can be derived using cyclotomic supplementarity instead: For any $G\in \mathcal G_{\mathbb Q} \setminus \mathcal G_{\mathbb D}$, there exists an odd prime number $p$ such that (SUP$_p$) is not provable in $\zx_G$: \[\zx_G \not\vdash  
\input{./figures/cyclo-supp.tikz}
\]

Hence the \zx-calculus needs to be completed to deal with rational angles. 
One possible way of doing this is to add the previous set of equations as axioms: $\mathrm{\Gamma}_{\frac {\pi} {4p}}(\Phi_{8p}) =  
\input{./figures/bra-0.tikz}
$. This would translate as: $$
\input{./figures/cyclo-as-axiom.tikz}
\qquad\text{with $p$ prime}$$ and -- as we will see in the following -- would be enough for completeness.
  However, instead of adding one or several new equations, we propose to add a simple and very natural rule to the language, the \emph{cancellation rule} which allows one to simplify non zero scalars: 
\begin{definition}[Cancellation rule]
The cancellation rule \integ is defined as such. 
For any diagrams of the \zx-Calculus $D_1$ and $D_2$:
 \hypertarget{r:integral}{}
\[ \forall\alpha\neq \pi\bmod2\pi,~~\zx\vdash D_1\otimes\gna=D_2\otimes\gna\underset{\textnormal{(cancel)}}{\implies}\zx\vdash D_1=D_2\]
\end{definition}

With this new rule \integ, the  equation $\mathrm{\Gamma}_{\frac {\pi} {4n}}(\Phi_{8n}) =  
\input{./figures/bra-0.tikz}
$  on cyclotomic polynomials is provable: 
\begin{lemma}
\label{lem:cyclo-to-0}For any $n>0$, 
$\zx_{\frac \pi{4n}}^{\textnormal{cancel}}\vdash  \mathrm{\Gamma}_{\frac {\pi} {4n}}(\Phi_{8n}) =  
\input{./figures/bra-0.tikz}
$
\end{lemma}

We show in the next subsection that the \zx-Calculus augmented with the new cancellation rule makes the \zx-calculus complete for rational angles.

\subsection{Normal forms}

First, let $G\in \mathcal G_{\mathbb Q} \setminus \mathcal G_{\mathbb D}$ be finite. Then, there exists $n$ such that $G$ is generated by $\frac{\pi}{4n}$ (i.e.~$G=\{\frac{k\pi}{4n}~|~k\in\mathbb{N}\}$), and for any $x$ in $\mathcal{R}_G$, there exists a polynomial $P\in\mathbb{D}[X]$ such that $x=P(e^{i\frac{\pi}{4n}})$.

This representation is not ideal. First of all, we can factor the powers of $\frac{1}{2}$ and write $P$ as $\frac{1}{2^p}Q$ where $Q\in\mathbb{Z}[X]$. The power $p$ can be uniquely chosen if we ensure that $Q$ is not a multiple of $2$ if $p>0$ i.e. $\forall Q'\in\mathbb{Z}[X]$, $p>0\implies Q\neq 2Q'$.

This expression is still not unique, because the evaluation of two different polynomials in $e^{i\frac{\pi}{4n}}$ can yield the same value (e.g.~$(e^{i\frac{\pi}{4n}})^{8n}=1$). To palliate this problem, we need to work in $\mathbb{Z}[X]/\phi_{8n}(X)$ where $\phi_{8n}$ is the $8n^{\text{th}}$ cyclotomic polynomial. Indeed, $\phi_{8n}$ is the unique irreducible polynomial with $e^{\frac{2i\pi}{8n}}$ as root. Then, applying the Euclidean division of $Q$ by $\phi_{8n}$:
\[\label{eq:div}
Q=Q'\phi_{8n}+R\tag{DIV}\] where $R$ and $Q'$ are uniquely chosen so that $\deg(R)<\deg(\phi_{8n})=\varphi(8n)$. Then, $Q(e^{i\frac{\pi}{4n}})=R(e^{i\frac{\pi}{4n}})$.

\begin{definition} 
Let $\Lambda_{\frac{\pi}{4n}}:\mathbb{N}\times \mathbb{Z}[X]\to \zx[1\to 0]$ be the map such that
$$\Lambda_{\frac{\pi}{4n}}(p,P):=
\input{./figures/control-x-pi_4.tikz}
$$
We then define $S_{\frac{\pi}{4n}}:=\left\lbrace\begin{array}{c|c}
\Lambda_{\frac{\pi}{4n}}(p,P) \quad&\quad \begin{array}{l}
P\in\mathbb{Z}[X],~
p\in\mathbb{N},\\
\deg(P)<\varphi(8n),\\
\forall Q\in\mathbb{Z}[X],~p>0\implies P\neq 2Q
\end{array}
\end{array}\right\rbrace$
\end{definition}

\begin{remark}
Notice that if $P=0$, only $\Lambda_{\frac{\pi}{4n}}(0,0)$ is part of $S_{\frac{\pi}{4n}}$. Indeed, if $P=0$, then $P=2\times 0=2P$, so the last constraint imposes that $p=0$.
\end{remark}
 
\begin{lemma}
$\interp{\Lambda_{\frac{\pi}{4n}}(p,P)}\ket{1}=\frac{1}{2^p}P(e^{i\frac{\pi}{4n}})$
\end{lemma}

Since every element of $\mathcal{R}_G$ is uniquely defined as $\frac{1}{2^p}P(e^{i\frac{\pi}{4n}})$ where $\deg(P)<\varphi(8n)$, and $\forall Q\in\mathbb{Z}[X],~p>0\implies P\neq 2Q$:
\begin{lemma}
The map $\eta_{\frac{\pi}{4n}}:S_{\frac{\pi}{4n}}\to \mathcal{R}_G = D\to \interp{D}\ket{1}$ is bijective.
\end{lemma}

We now need to meet the conditions of Theorem \ref{them:completeness}. First we notice that we can operate the sum and the product on controlled polynomials:

\begin{lemma}
\label{lem:sum-prod-polynomials}
For any polynomials $P$ and $Q$ in $\mathbb{Z}[X]$:
\[\zx_{\frac{\pi}{4n}}\vdash\left(\scalebox{0.95}{
\input{./figures/sum-of-control-polynomials.tikz}
}\right),\left(\scalebox{0.95}{
\input{./figures/prod-of-control-polynomials.tikz}
}\right)\]
\end{lemma}

Two problems arise when trying to do the same with diagrams of $S_{\frac{\pi}{4n}}$. First of all, the sum of two diagrams in normal form can have a parity issue. For instance $\frac{1}{2}(2+X)+\frac{1}{2}(X+2X^2)=\frac{1}{2}(2+2X+2X^2)$ which shall be reduced to $1+X+X^2$. This is dealt with thanks to the following lemma:
\begin{lemma}
\[\zx_{\frac{\pi}{4n}}\vdash
\input{./figures/C-half-times-C-2P.tikz}
\]
\end{lemma}

Secondly, the product of two polynomials may well end up with a degree larger than $\varphi(8n)$. However, since we can operate the sum and product of controlled polynomials thanks to Lemma \ref{lem:sum-prod-polynomials}, we can derive the controlled version of the Euclidean division (\ref{eq:div}). Combined with Lemma \ref{lem:cyclo-to-0}, we get, assuming $P= Q\phi_{8n}+R$:
\def\fig{control-polynomial-division}
\begin{align*}
\zx_{\frac{\pi}{4n}}^{\textnormal{cancel}}\vdash~~\begin{tikzpicture}
	\begin{pgfonlayer}{nodelayer}
		\node [style=gn] (0)  at (-2.375, 0.75) {$\frac{k\pi}{2^n}$};
		\node [style=gn, align=center] (1)  at (-1.375, 0.75) {$\frac{2p\text{-}1}{2^{n\text{-}m}}\pi$\\${+}\pi$};
		\node [style=gn] (2)  at (-0.375, 0.0) {$\frac{(2p\text{-}1)\pi}{2^{n\text{-}m\text{-}1}}$};
		\node [style=gn] (3)  at (0.875, 0.0) {$\frac{(2p\text{-}1)\pi}{2^{n\text{-}m\text{-}2}}$};
		\node [style=gn] (4)  at (2.375, 0.0) {$\frac{2p\text{-}1}{2}\pi$};
		\node [style=none] (5)  at (1.625, 0.0) {$~\cdots~$};
		\node [style=rn] (21)  at (-1.625, -0.75) {};
		\node [style=gn] (22)  at (-1.625, -0.25) {};
		\node [style=gn] (23)  at (-2.125, -0.25) {};
		\node [style=rn] (24)  at (-2.125, -0.75) {};
	\end{pgfonlayer}
	\begin{pgfonlayer}{edgelayer}
		\draw [style=none, bend right=45, looseness=1.25] (21) to (22);
		\draw [style=none] (21) to (22);
		\draw [style=none, bend right=45, looseness=1.25] (22) to (21);
		\draw [style=none, bend right=45, looseness=1.25] (23) to (24);
		\draw [style=none, bend right=45, looseness=1.25] (24) to (23);
		\draw [style=none] (24) to (23);
	\end{pgfonlayer}
\end{tikzpicture}
\eq{\ref{lem:sum-prod-polynomials}\\\text{(\ref{eq:div})}}\begin{tikzpicture}
	\begin{pgfonlayer}{nodelayer}
		\node [style=uglabel] (7)  at (0.625, -0.875) {$\frac{2p-1}{2}\pi$};
		\node [style=none, yshift=6pt] (8)  at (0.375, -0.625) {$~\vdots~$};
		\node [style=gn] (9)  at (0.375, -0.125) {};
		\node [style=uglabel] (10)  at (0.625, 0.375) {$\frac{2p-1}{2^{n-m-1}}\pi$};
		\node [style=uglabel] (11)  at (0.625, -0.125) {$\frac{2p-1}{2^{n-m-2}}\pi$};
		\node [style=gn] (12)  at (0.375, 0.375) {};
		\node [style=gn] (13)  at (-0.625, 0.875) {};
		\node [style=uglabel] (14)  at (-0.375, 0.875) {$\frac{2p-1}{2^{n-m-1}}\pi{+}\pi$};
		\node [style=gn] (15)  at (0.375, -0.875) {};
		\node [style=rn] (25)  at (-0.125, -0.875) {};
		\node [style=gn] (26)  at (-0.125, -0.375) {};
		\node [style=gn] (27)  at (-0.625, -0.375) {};
		\node [style=rn] (28)  at (-0.625, -0.875) {};
	\end{pgfonlayer}
	\begin{pgfonlayer}{edgelayer}
		\draw [style=none, bend right=45, looseness=1.25] (25) to (26);
		\draw [style=none] (25) to (26);
		\draw [style=none, bend right=45, looseness=1.25] (26) to (25);
		\draw [style=none, bend right=45, looseness=1.25] (27) to (28);
		\draw [style=none, bend right=45, looseness=1.25] (28) to (27);
		\draw [style=none] (28) to (27);
	\end{pgfonlayer}
\end{tikzpicture}
\eq{\ref{lem:cyclo-to-0}}
\eq{}\cdots\eq{}\begin{tikzpicture}
	\begin{pgfonlayer}{nodelayer}
		\node [style=gn] (34)  at (0.25, 0.0) {};
		\node [style=gn] (35)  at (0.0, 0.5) {};
		\node [style=gn] (36)  at (-0.25, 0.0) {};
		\node [style=rn] (37)  at (0.25, -0.5) {};
		\node [style=rn] (38)  at (-0.25, -0.5) {};
	\end{pgfonlayer}
	\begin{pgfonlayer}{edgelayer}
		\draw [style=none, bend right=45, looseness=1.25] (34) to (37);
		\draw [style=none, bend right=45, looseness=1.25] (36) to (38);
		\draw [style=none, bend right=45, looseness=1.25] (37) to (34);
		\draw [style=none] (37) to (34);
		\draw [style=none, bend right=45, looseness=1.25] (38) to (36);
		\draw [style=none] (38) to (36);
	\end{pgfonlayer}
\end{tikzpicture}
\end{align*}

All in all, any controlled scalar in the form $\Lambda_{\frac{\pi}{4n}}P$ can be reduced to a diagram in $S_{\frac{\pi}{4n}}$.
\begin{lemma}
\[\zx_{\frac{\pi}{4n}}^{\textnormal {cancel}}\vdash\left(
\input{./figures/condition-on-exp-i-alpha-d-q.tikz}
\right)\!,\!\left(
\input{./figures/prod-of-control-scalars-d-q.tikz}
\right)\!,\!\left(
\input{./figures/sum-of-control-scalars-d-q.tikz}
\right)\]
\end{lemma}

\begin{theorem}
\label{thm:pi_4n-completeness}
The \frag{4n} of the \zx-Calculus with set of rules $\zx_{\frac{\pi}{4n}}^{\textnormal {cancel}}$ is complete, and any ZX-diagram can be put into a normal form with respect to $S_{\frac{\pi}{4n}}$.
\end{theorem}

\begin{corollary}
\label{cor:qpi-completeness}
For any $G\in \mathcal G_{\mathbb Q}$ (finite or not), the fragment $G$ with set of rules $\zx_G^{\textnormal {cancel}}$ is complete, and any \zx-diagram can be put into a normal form with respect to $S_G:= \bigcup\limits_{\frac{\pi}{4n} \in G} S_{\frac{\pi}{4n}}$.
\end{corollary}

\section{Normal Forms with dyadic angles}
\label{sec:pi_2^n-fragments}

In this section we focus on a particular case of dyadic angles, i.e. a subgroup of $\mathbb D\pi$ which contains $\frac \pi4$ (i.e. $G\in \mathcal G_{\mathbb D}$). 
 In the previous section, we introduced the cancellation rule which makes the \zx-calculus complete for rational angles. 

Notice that, given a fragment $G\in \mathcal G$, the cancellation rule can be derived from the other rules if for every $\alpha\in G$, $\alpha\neq 0\bmod \pi$, there exists an inverse of \begin{tikzpicture}
	\begin{pgfonlayer}{nodelayer}
		\node [style=gn] (0) at (0, -0) {$\alpha$};
	\end{pgfonlayer}
\end{tikzpicture}, i.e. a diagram $D\in \zx_G[0\to 0]$ s.t. $\interp{D  \otimes \begin{tikzpicture}
	\begin{pgfonlayer}{nodelayer}
		\node [style=gn] (0) at (0, -0) {$\alpha$};
	\end{pgfonlayer}
\end{tikzpicture}} = 1$, and moreover this equation is provable:  $\zx_G\vdash D  \otimes \begin{tikzpicture}
	\begin{pgfonlayer}{nodelayer}
		\node [style=gn] (0) at (0, -0) {$\alpha$};
	\end{pgfonlayer}
\end{tikzpicture} = \scalebox{0.7}{
\input{./figures/empty-diagram.tikz}
}$. This is the case in any fragment of dyadic angles:

\begin{lemma}
\label{lem:pi_2^n-inverse}
For any $n\geq 1$, and any $k\in\{-2^n+1,\cdots, 2^{n+1}-1\}$, \begin{tikzpicture}
	\begin{pgfonlayer}{nodelayer}
		\node [style=gn] (0) at (0, -0) {$\frac{k\pi}{2^n}$};
	\end{pgfonlayer}
\end{tikzpicture} has an inverse. There exist $0\leq m<n$ and $p\in\mathbb{Z}$ such that:
\def\fig{gn-k-pi_2-to-n-inverse}
\[\eq{}\begin{tikzpicture}
	\begin{pgfonlayer}{nodelayer}
		\node [style=none] (40)  at (-0.25, -0.25) {};
		\node [style=none] (41)  at (-0.25, 0.25) {};
		\node [style=none] (42)  at (0.25, 0.25) {};
		\node [style=none] (43)  at (0.25, -0.25) {};
	\end{pgfonlayer}
	\begin{pgfonlayer}{edgelayer}
		\draw [style=dashed] (40.center) to (43.center);
		\draw [style=dashed] (41.center) to (40.center);
		\draw [style=dashed] (42.center) to (41.center);
		\draw [style=dashed] (43.center) to (42.center);
	\end{pgfonlayer}
\end{tikzpicture}\]
\end{lemma}

\begin{theorem}
\label{thm:pi_2-to-n-completeness}
For $n\geq 2$, the \frag{2^n} of the ZX-Calculus with set of rules $\zx_{\frac{\pi}{2^n}}$ is complete, and any ZX-diagram can be put into a normal form with respect to $S_{\frac{\pi}{2^n}}=S_{\frac{\pi}{4\times 2^{n\text{-}2}}}$.
\end{theorem}

\begin{corollary}
\label{cor:dpi-completeness}
For any $G\in \mathcal G_{\mathbb D}$ (finite or not), the fragment $G$ with set of rules $\zx_G$ is complete, and any ZX-diagram can be put into a normal form with respect to $S_G:= \bigcup\limits_{\frac{\pi}{2^n} \in G} S_{\frac{\pi}{2^n}}$.
\end{corollary}

\section{Discussion}
\label{sec:discussion}

We now have a constructive proof for the completeness of the \frag4 and of the general ZX-Calculus. Additionally, we used the ``generic'' normal form to prove the completeness of the \frag{2^n} for any $n$. When $n\geq2$, the \frag{2^n} uses the set of rules ZX. In the general case, \add is added to this set, and it has been proven to be necessary. We remind the complete axiomatisations used for the different fragments reviewed in this article in Figure \ref{fig:fragments-axiomatisations}.

\begin{figure}[htb]
\begin{center}
\def\arraystretch{1.5}
\begin{tabular}{|@{~~}c@{~~}|@{~~}c@{~~}|}
\hline
Fragments & Complete Axiomatisation\\
\hline\hline
$G\subseteq \mathbb{D}\pi$ & ZX$_{G}$\\
\hline
$G\subseteq \mathbb{Q}\pi$ & ZX$_{G}^{\textnormal{cancel}}$\\
\hline
General & ZX$^{\textnormal A}$\\
\hline
\end{tabular}
\end{center}
\caption[]{Considered fragments and their complete axiomatisations}\label{fig:fragments-axiomatisations}
\end{figure}

We proved that any \frag{4n} is complete with the set of axioms ZX augmented with the meta-rule \integ. We leave as open the existence of a set of axioms that makes the \frag{4n}s complete without the use of a \emph{meta}-rule. Such a potential (family of) axiom(s) has been identified as the cyclotomic supplementarity \cite{gen-supp}.
This can provide the inverse of \gna~but only for some values of $\alpha$. For instance, in the \frag{12}: 
\def\fig{gn-2pi_3-inverse}
\begin{align*}

\eq{}
\eq{\text{(SUP$_3$)}}
\eq{}
\end{align*}

Finally, it is to be noticed that all the fragments considered in this paper contains the angle $\frac \pi4$, as some axioms of \ref{fig:ZX_rules} contains $\frac \pi 4$. However the results presented in this paper can be generalised to fragments which do not contain $\frac \pi 4$ using the $\mathrm{\Delta}\zx$ \cite{zx-toffoli}, where the triangle is part of the syntax.

\section*{Acknowledgement}
The authors acknowledge support from the projects ANR-17-CE25-0009 SoftQPro, ANR-17-CE24-0035 VanQuTe, PIA-GDN/Quantex, and STIC-AmSud 16-STIC-05 FoQCoSS. SP is also supported by the LUE project UOQ. Most diagrams were written with the help of TikZit.

\appendix
\section{Appendix}

\subsection{Dirac Notation and the Choi-Jamio\l kowski Isomorphism}

The state of a qubit is a vector in the 2 dimensional Hilbert space $\mathbb C^{2^n}$.  $\ket 0 := {1 \choose 0}$ and $\ket 1 := {0 \choose 1}$ form the so-called standard basis of  $\mathbb C^{2^n}$. We use the notation $\ket + := \frac{\ket 0+\ket 1}{\sqrt 2}$ and  $\ket - := \frac{\ket 0-\ket 1}{\sqrt 2}$. The two states $\ket +, \ket -$ also form a basis, the so-called diagonal basis. Any qubit state $\ket \phi \in \mathbb C^{2}$ can be decomposed in the standard basis: $$\ket \phi = \alpha\ket 0 +\beta \ket 1$$ 

A valid quantum state is normalised, i.e. $|\alpha|^2+|\beta|^2=1$. Notice that in the present paper we consider normalised but also unnormalised quantum states. 

More generally, the state of a register of $n$ qubits is a vector $\ket \phi \in \mathbb C^{2^n}$. For any $x=x_0\ldots x_{n-1}\in \{0,1\}^n$, let $\ket x:= \ket{x_0}\otimes \ldots \otimes \ket{x_{n-1}}$ where $\cdot\otimes\cdot$ is the Kronecker product, in other words $\ket x$ is a vector in which all entries are $0$ except the entry number $\sum_{i=0}^{n-1}x_i2^{n-1-i}$, which is $1$. Similarly, $\ket {+^n}  :=\ket {+}\otimes \ldots\otimes  \ket +$ and $\ket {-^n}=\ket {-}\otimes \ldots \otimes\ket -$. The set of  states $\{\ket x ~|~x\in \{0,1\}^n\}$ forms a  basis,  thus any $n$-qubit state $\ket \phi \in \mathbb C^{2^n}$ can be described as $$\ket \phi = \sum_{x\in \{0,1\}^n}\alpha_x\ket x$$

Notice that a zero-qubit state (when $n=0$) is a scalar i.e.~an element of $\mathbb C$. We define $\ket \epsilon = 1$, where $\epsilon$ denotes the empty word. 

The adjoint of a state $\ket \phi = \sum_{x\in \{0,1\}^n}\alpha_x\ket x\in \mathbb C^{2^n} $ is $\bra \phi  :=(\ket{\phi})^\dagger =\sum_{x\in \{0,1\}^n}\alpha^
*_x\bra x\in \mathbb C^{2^n} \to 1$, where $\forall x\in \{0,1\}^n$ $\bra x$ is the unique linear map such that $\forall y\in \{0,1\}^n$, $\bra x \ket y = \delta_{x,y}$.  

Given a linear map $M:\mathbb C^{2^n} \to \mathbb C^{2^m}$, for any $x\in \{0,1\}^n$, $M$ maps $\ket x$ to $M\ket x = \sum_{y\in \{0,1\}^m}\alpha_{x,y}\ket y$. Using the Dirac notation, $M$ can be represented as follows:
$$M = \sum_{x\in \{0,1\}^n,y\in \{0,1\}^m}\alpha_{x,y} \ket y\!\bra x$$

The Choi-Jamio\l kowski isomorphism, or state/map duality is the following isomorphism between linear maps and quantum states: $$\sum_{x\in \{0,1\}^n,y\in \{0,1\}^m}\alpha_{x,y} \ket y\!\bra x \qquad \mapsto\quad \sum_{x\in \{0,1\}^n,y\in \{0,1\}^m}\alpha_{x,y} \ket {xy}$$

\subsection{Already Proven Lemmas}

\begin{multicols}{2}

\begin{lemma}
\label{lem:inverse}
\[
\input{./figures/inverse.tikz}
\]
\end{lemma}

\begin{lemma}
\label{lem:multiplying-global-phases}
\[
\input{./figures/multiplying-global-phases.tikz}
\]
\end{lemma}

\begin{lemma}
\label{lem:bicolor-0-alpha}
\[
\input{./figures/bicolor-0-alpha.tikz}
\]
\end{lemma}

\begin{lemma}
\label{lem:hopf}
\[
\input{./figures/hopf.tikz}
\]
\end{lemma}

\begin{lemma}
\label{lem:k1}
\[
\input{./figures/k1.tikz}
\]
\end{lemma}

\begin{lemma}
\label{lem:h-loop}
\[
\input{./figures/h-loop.tikz}
\]
\end{lemma}

\begin{lemma}
\label{lem:supp-to-minus-pi_4}
\[
\input{./figures/supp-to-minus-pi_4.tikz}
\]
\end{lemma}

\begin{lemma}
\label{lem:green-state-pi_2-is-red-state-minus-pi_2}
\[
\input{./figures/green-state-pi_2-is-red-state-minus-pi_2.tikz}
\]
\end{lemma}

\begin{lemma}
\label{lem:euler-decomp-with-scalar}
\[
\input{./figures/euler-decomp-with-scalar.tikz}
\]
\end{lemma}

\begin{lemma}
\label{lem:triangle-hadamard-parallel}
\[
\input{./figures/triangle-hadamard-parallel.tikz}
\]
\end{lemma}

\begin{lemma}
\label{lem:looped-triangle}
\[
\input{./figures/looped-triangle.tikz}
\]
\end{lemma}

\begin{lemma}
\label{lem:black-dot-swappable-outputs}
\[
\input{./figures/lemma-black-dot-swappable-outputs.tikz}
\]
\end{lemma}

\begin{lemma}
\label{lem:C1-original}
\[
\input{./figures/control-commutation-2.tikz}
\]
\end{lemma}

\begin{lemma}
\label{lem:red-state-on-triangle}
\[
\input{./figures/red-state-on-triangle.tikz}
\]
\end{lemma}

\begin{lemma}
\label{lem:pi-red-state-on-triangle}
\[
\input{./figures/pi-red-state-on-triangle.tikz}
\]
\end{lemma}

\begin{lemma}
\label{lem:red-state-on-upside-down-triangle}
\[
\input{./figures/red-state-on-upside-down-triangle.tikz}
\]
\end{lemma}

\begin{lemma}
\label{lem:pi-red-state-on-upside-down-triangle}
\[
\input{./figures/pi-red-state-on-upside-down-triangle.tikz}
\]
\end{lemma}

\begin{lemma}
\label{lem:pi-green-state-on-upside-down-triangle}
\[
\input{./figures/pi-green-state-on-upside-down-triangle.tikz}
\]
\end{lemma}

\begin{lemma}
\label{lem:not-triangle-is-symmetrical}
\[
\input{./figures/not-ug-is-symmetrical.tikz}
\]
\end{lemma}

\begin{lemma}
\label{lem:symmetric-diagram-with-triangle-hadamard}
\[
\input{./figures/symmetric-diagram-with-triangle-hadamard.tikz}
\]
\end{lemma}

\begin{lemma}
\label{lem:bw-triangle}
\[
\input{./figures/bw-triangle.tikz}
\]
\end{lemma}

\begin{lemma}
\label{lem:inverse-of-triangle}
\[
\input{./figures/inverse-of-triangle.tikz}
\]
\end{lemma}

\begin{lemma}
\label{lem:triangle-through-W}
\[
\input{./figures/triangle-through-W.tikz}
\]
\end{lemma}

\begin{lemma}
\label{lem:triangles-fork-absorbs-anti-CNOT}
\[
\input{./figures/ug-fork-absorbs-anti-CNOT.tikz}
\]and\[
\input{./figures/ug-fork-absorbs-CNOT.tikz}
\]
\end{lemma}

\begin{lemma}
\label{lem:parallel-triangles}
\[
\input{./figures/parallel-triangles.tikz}
\]
\end{lemma}

\begin{lemma}
\label{lem:old-bw}
\[
\input{./figures/2-diagrams-of-control-triangle-axiom-simplified.tikz}
\]
\end{lemma}

\end{multicols}

\subsection{New Lemmas with ZX}

\begin{multicols}{2}

\begin{lemma}
\def\fig{control-alpha-with-triangles}
\label{lem:control-alpha-with-triangles}
\[\eq{}\input{./figures/\fig/\fig_13.tikz}\]
\end{lemma}

\begin{lemma}
\def\fig{2Id-is-C2-times-anti-C2}
\label{lem:2Id-is-C2-times-anti-C2}
\[\eq{}\]
\end{lemma}

\begin{lemma}
\def\fig{2-triangle-cycle-proof}
\label{lem:2-triangle-cycle}
\[\eq{}\input{./figures/\fig/\fig_14.tikz}\]
\end{lemma}

\begin{lemma}
\def\fig{2-cycle-triangle-with-one-not-proof}
\label{lem:2-cycle-triangle-with-one-not}
\[\eq{}\input{./figures/\fig/\fig_05.tikz}\]
\end{lemma}

\begin{lemma}
\def\fig{2-through-triangle}
\label{lem:2-through-triangle}
\[\eq{}\]
\end{lemma}

\begin{lemma}
\def\fig{CNOT-under-controlsep}
\label{lem:CNOT-under-controlsep}
\[\eq{}\input{./figures/\fig/\fig_05.tikz}\]
\end{lemma}

\begin{lemma}
\def\fig{swapped-CNOT-under-controlsep}
\label{lem:swapped-CNOT-under-controlsep}
\[\eq{}\]
\end{lemma}

\begin{lemma}
\def\fig{controlsep-swapped-wires}
\label{lem:controlsep-swapped-wires}
\[\eq{}\input{./figures/\fig/\fig_05.tikz}\]
\end{lemma}

\begin{lemma}
\def\fig{controlsep-surrounded-by-h}
\label{lem:controlsep-surrounded-by-h}
\[\eq{}\input{./figures/\fig/\fig_08.tikz}\]
\end{lemma}

\begin{lemma}
\def\fig{two-consecutive-controlsep-can-share-CNOT}
\label{lem:two-consecutive-controlsep-can-share-CNOT}
\[\eq{}\]
\end{lemma}

\begin{lemma}
\def\fig{two-consecutive-controlsep-can-be-swapped}
\label{lem:two-consecutive-controlsep-can-be-swapped}
\[\eq{}\input{./figures/\fig/\fig_05.tikz}\]
\end{lemma}

\begin{lemma}
\def\fig{red-state-on-controlsep}
\label{lem:red-state-on-controlsep}
\[\eq{}\]
and
\[\input{./figures/\fig/\fig_05.tikz}\eq{}\input{./figures/\fig/\fig_08.tikz}\]
\end{lemma}

\begin{lemma}
\def\fig{anti-CNOT-on-controlsep}
\label{lem:anti-CNOT-on-controlsep}
\[\eq{}\]
\end{lemma}

\begin{lemma}
\def\fig{two-consecutive-controlsep-same-control}
\label{lem:two-consecutive-controlsep-same-control}
\[\eq{}\]
\end{lemma}

\begin{lemma}
\def\fig{two-consecutive-controlsep-anticontrol}
\label{lem:two-consecutive-controlsep-anticontrol}
\[\eq{}\]
\end{lemma}

\begin{lemma}
\def\fig{controlsep-crosses-green-node}
\label{lem:controlsep-crosses-green-node}
\[\eq{}\input{./figures/\fig/\fig_07.tikz}\]
\end{lemma}

\begin{lemma}
\def\fig{anti-controlsep}
\label{lem:anti-controlsep}
\[\eq{}\]
\end{lemma}

\begin{lemma}
\def\fig{controlsep-and-control-alpha}
\label{lem:controlsep-and-control-alpha}
\[\eq{}\input{./figures/\fig/\fig_06.tikz}\]
\end{lemma}

\begin{lemma}
\def\fig{alpha-through-anti-controlsep}
\label{lem:alpha-through-(anti)-controlsep}
\[\eq{}\input{./figures/\fig/\fig_05.tikz}\]
\end{lemma}

\begin{lemma}
\def\fig{gn-distrib-over-W}
\label{lem:gn-distrib-over-W}
\[\eq{}\input{./figures/\fig/\fig_05.tikz}\]
\end{lemma}

\end{multicols}

\begin{lemma}
\label{lem:cycle-triangle}
\[
\input{./figures/quasi-triangle-cycle.tikz}
\]
and
\[
\input{./figures/triangle-cycle-pure.tikz}
\]
\end{lemma}

\begin{lemma}
\label{lem:other-forms-of-concatenation-gadget}
\def\fig{other-forms-of-concatenation-gadget}
\[
\eq{}
\eq{}\]
\end{lemma}

\begin{multicols}{2}

\begin{lemma}
\def\fig{concatenation-gadget-NOT}
\label{lem:concatenation-gadget-NOT}
\[\eq{}\]
\end{lemma}

\begin{lemma}
\def\fig{concatenation-gadget-r-0}
\label{lem:concatenation-gadget-r-0}
\[\eq{}\]
\end{lemma}

\begin{lemma}
\def\fig{concatenation-gadget-l-0}
\label{lem:concatenation-gadget-l-0}
\[\eq{}\]
\end{lemma}

\begin{lemma}
\def\fig{C2-distributed-through-W-2}
\label{lem:C2-distributed-through-W-2}
\[\eq{}\]
and 
\def\fig{C-half-distributed-through-W}
\[\eq{}\]
\end{lemma}

\begin{lemma}
\def\fig{Calpha-through-W}
\label{lem:Calpha-through-W}
\[\eq{}\input{./figures/\fig/\fig_17.tikz}\]
\end{lemma}

\begin{lemma}
\def\fig{Cx-through-anti-controlsep}
\label{lem:Cx-through-(anti)-controlsep}
\[\eq{}\]
\end{lemma}

\begin{lemma}
\def\fig{looped-controlsep}
\label{lem:looped-controlsep}
\[\eq{}\quad\text{and}\quad\eq{}\input{./figures/\fig/\fig_07.tikz}\]
\end{lemma}

\begin{lemma}
\def\fig{CUP-on-anti-controlsep}
\label{lem:CUP-on-(anti)-controlsep}
\[\eq{}\input{./figures/\fig/\fig_08.tikz}\]
\end{lemma}
\end{multicols}

\begin{lemma}
\def\fig{gn-1-0-on-control-state-intermed-lemma-bis}
\label{lem:gn-1-0-on-control-state-intermed-lemma}
\[\eq{}\input{./figures/\fig/\fig_09.tikz}\]
\end{lemma}

\begin{multicols}{2}
\begin{lemma}
\def\fig{anti-controlsep-under-addition}
\label{lem:(anti)-controlsep-under-addition}
\[\eq{}\input{./figures/\fig/\fig_06.tikz}\]
\end{lemma}

\begin{lemma}
\def\fig{cos-alpha}
\label{lem:cos-alpha}
\[\eq{}\input{./figures/\fig/\fig_05.tikz}\]
\end{lemma}

\begin{lemma}
\def\fig{anti-controlsep-triangle}
\label{lem:(anti)-controlsep-triangle}
\[\eq{}\]
\end{lemma}

\begin{lemma}
\label{lem:commutation-base-component}
\def\fig{base-building-component-commutation}
\begin{align*}

\eq{}
\end{align*}
\end{lemma}

\begin{lemma}
\label{lem:polynomial-through-triangle}
The following result is not only true for any $\Gamma_{\alpha}P$ but for any finite sum of them. Hence we extend the previously defined $\Gamma$ as:
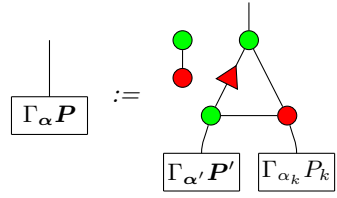
\[
\input{./figures/multivariate-sum-of-control-polynomials.tikz}
\]
where $\vec \alpha = \alpha_0,\cdots,\alpha_k$, $\vec{\alpha'}= \alpha_0,\cdots,\alpha_{k\text{-}1}$, $\vec P = P_0,\cdots,P_k$, $\vec{P'} = P_0,\cdots,P_{k\text{-}1}$. Then:
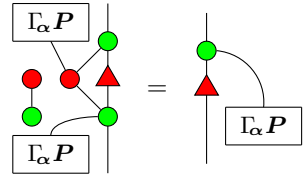
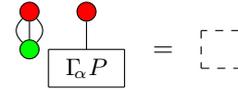
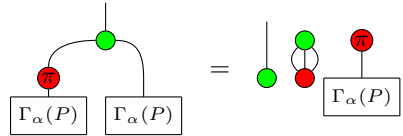
\def\fig{polynomial-through-triangle}
\[\eq{}\input{./figures/\fig/\fig_12.tikz}\]
\end{lemma}

\begin{lemma}
\label{lem:red-state-on-polynomial}
\def\fig{red-state-on-polynomial}
\[\eq{}\input{./figures/\fig/\fig_05.tikz}\]
\end{lemma}

\begin{lemma}
\label{lem:polynomial-times-not-polynomial}
\def\fig{polynomial-times-not-polynomial}
\[\eq{}\]
\end{lemma}


\end{multicols}

\subsection{Lemmas with ZX$^{\text A}$}

\begin{multicols}{2}

\begin{lemma}
\label{lem:add-axiom-triangle}
\def\fig{add-axiom-triangle}
\[\input{./figures/\fig/\fig_05.tikz}\eq{}\]
with $2e^{i\theta_3}\cos(\gamma)=e^{i\theta_1}\cos(\alpha)+e^{i\theta_2}\cos(\beta)$.
\end{lemma}

\begin{lemma}
\label{lem:prod-cos}
\def\fig{prod-cos}
\[\eq{}\input{./figures/\fig/\fig_09.tikz}\]
with $\cos(\gamma) = \cos(\alpha)\cos(\beta)$.
\end{lemma}

\begin{lemma}
\label{lem:add-bis}
We can deduce an equality similar to the rule \add:
\[
\input{./figures/add-axiom-bis.tikz}
\]
\end{lemma}

\begin{lemma}
\label{lem:ctrl-power-2}
Let $\rho\in\mathbb{R}+$. Then, for any $n_1,n_2\geq\max\left(0,\left\lceil \log_2(\rho)\right\rceil\right)$:
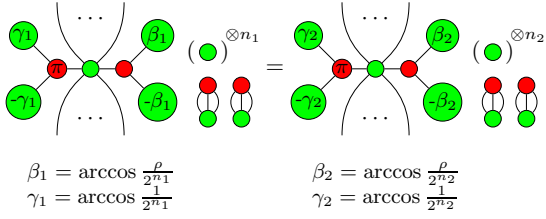
\[\fit{
\input{./figures/ctrl-power-2.tikz}
}\]
\end{lemma}

\begin{corollary}
\label{cor:arccos-2-power-n}
For any $n\in\mathbb{N}$, with $\gamma=\arccos{\frac{1}{2^n}}$:
\def\fig{lemma-arccos-1_2-to-the-n-branches}
\begin{align*}

\eq{}
\end{align*}
\end{corollary}

\end{multicols}

\subsection{Proof of the New Lemmas}

\begin{proof}[Lemma \ref{lem:control-alpha-with-triangles}]
\def\fig{control-alpha-with-triangles}
\begin{align*}
&
\eq{\ref{lem:supp-to-minus-pi_4}}
\eq{\soo\\\ref{lem:euler-decomp-with-scalar}}
\eq{\h\\\bt\\\soo}\\&
\eq{\soo\\\h\\\bt}
\eq{\ref{lem:C1-original}}\input{./figures/\fig/\fig_05.tikz}\\&
\eq{\bo\\\stt\\\soo\\\h\\\ref{lem:k1}}\input{./figures/\fig/\fig_06.tikz}
\eq{\bt}\input{./figures/\fig/\fig_07.tikz}
\eq{\h\\\soo\\\ref{lem:h-loop}}\input{./figures/\fig/\fig_08.tikz}\\&
\eq{\ref{lem:C1-original}}\input{./figures/\fig/\fig_09.tikz}
\eq{\h\\\ref{lem:h-loop}\\\soo}\input{./figures/\fig/\fig_10.tikz}
\eq{\bt\\\soo}\input{./figures/\fig/\fig_11.tikz}\\&
\eq{\bo\\\soo}\input{./figures/\fig/\fig_12.tikz}
\eq{\ref{lem:euler-decomp-with-scalar}\\\bo\\\soo}\input{./figures/\fig/\fig_13.tikz}
\end{align*}
\qed\end{proof}

\begin{proof}[Lemma \ref{lem:2Id-is-C2-times-anti-C2}]
\def\fig{2Id-is-C2-times-anti-C2}
\begin{align*}

\eq{\bo}
\eq{\soo\\\ref{lem:not-triangle-is-symmetrical}\\\stt}
\eq{\ref{lem:control-alpha-with-triangles}}
\eq{\bo}
\end{align*}
\qed\end{proof}

\begin{proof}[Lemma \ref{lem:2-triangle-cycle}]
\def\fig{2-triangle-cycle-proof}
\begin{align*}
&
\eq{\stt\\\soo\\\ref{lem:not-triangle-is-symmetrical}}
\eq{\soo\\\ref{lem:supp-to-minus-pi_4}}
\eq{\ref{lem:euler-decomp-with-scalar}\\\soo}\\&
\eq{\bt}
\eq{\eu}\input{./figures/\fig/\fig_05.tikz}
\eq{\h\\\soo}\input{./figures/\fig/\fig_06.tikz}
\eq{\ref{lem:C1-original}}\input{./figures/\fig/\fig_07.tikz}\\&
\eq{\h}\input{./figures/\fig/\fig_08.tikz}
\eq{\ref{lem:k1}\\\ref{lem:euler-decomp-with-scalar}\\\h\\\soo}\input{./figures/\fig/\fig_09.tikz}
\eq{\soo}\input{./figures/\fig/\fig_10.tikz}
\eq{\ref{lem:h-loop}\\\soo}\input{./figures/\fig/\fig_11.tikz}\\&
\eq{\bt}\input{./figures/\fig/\fig_12.tikz}
\eq{\ref{lem:green-state-pi_2-is-red-state-minus-pi_2}}\input{./figures/\fig/\fig_13.tikz}
\eq{\soo\\\stt}\input{./figures/\fig/\fig_14.tikz}
\end{align*}
\qed\end{proof}

\begin{proof}[Lemma \ref{lem:2-cycle-triangle-with-one-not}]
\def\fig{2-cycle-triangle-with-one-not-proof}
\begin{align*}

\eq{\ref{lem:looped-triangle}}
\eq{\ref{lem:black-dot-swappable-outputs}}
\eq{\ref{lem:2-triangle-cycle}}
\eq{\soo}
\eq{\ref{lem:hopf}}\input{./figures/\fig/\fig_05.tikz}
\end{align*}
\qed\end{proof}

\begin{proof}[Lemma \ref{lem:2-through-triangle}]
\def\fig{2-through-triangle}
\begin{align*}

\eq{\stt\\\soo\\\ref{lem:not-triangle-is-symmetrical}\\\ref{lem:k1}\\\h}
\eq{\ref{lem:symmetric-diagram-with-triangle-hadamard}\\\stt\\\ref{lem:2Id-is-C2-times-anti-C2}}
\eq{\ref{lem:symmetric-diagram-with-triangle-hadamard}}
\eq{\ref{lem:bw-triangle}\\\soo\\\stt}
\end{align*}
\qed\end{proof}

\begin{proof}[Lemma \ref{lem:CNOT-under-controlsep}]
\def\fig{CNOT-under-controlsep}
\begin{align*}
&
\eq{}
\eq{\bt\\\soo}
\eq{\ref{lem:triangle-through-W}\\\soo}\\&
\eq{\ref{lem:hopf}\\\stt}
\eq{}\input{./figures/\fig/\fig_05.tikz}
\end{align*}
\qed\end{proof}

\begin{proof}[Lemma \ref{lem:swapped-CNOT-under-controlsep}]
\def\fig{swapped-CNOT-under-controlsep}
\begin{align*}

\eq{}
\eq{\ref{lem:looped-triangle}}
\eq{}
\end{align*}
\qed\end{proof}

\begin{proof}[Lemma \ref{lem:controlsep-swapped-wires}]
\def\fig{controlsep-swapped-wires}
\begin{align*}
&
\eq{\bt\\\ref{lem:hopf}}
\eq{\ref{lem:CNOT-under-controlsep}}
\eq{\ref{lem:swapped-CNOT-under-controlsep}}\\&
\eq{\ref{lem:k1}}
\eq{\ref{lem:CNOT-under-controlsep}}\input{./figures/\fig/\fig_05.tikz}
\end{align*}
\qed\end{proof}

\begin{proof}[Lemma \ref{lem:controlsep-surrounded-by-h}]
\def\fig{controlsep-surrounded-by-h}
\begin{align*}

\eq{\ref{lem:controlsep-swapped-wires}}
\eq{}
\eq{\ref{lem:not-triangle-is-symmetrical}\\\h}
\eq{\ref{lem:symmetric-diagram-with-triangle-hadamard}}\\
\eq{\h}\input{./figures/\fig/\fig_05.tikz}
\eq{\ref{lem:symmetric-diagram-with-triangle-hadamard}}\input{./figures/\fig/\fig_06.tikz}
\eq{\ref{lem:not-triangle-is-symmetrical}}\input{./figures/\fig/\fig_07.tikz}
\eq{}\input{./figures/\fig/\fig_08.tikz}
\end{align*}
\qed\end{proof}

\begin{proof}[Lemma \ref{lem:two-consecutive-controlsep-can-share-CNOT}]
\def\fig{two-consecutive-controlsep-can-share-CNOT}
\begin{align*}

\eq{\ref{lem:controlsep-swapped-wires}}
\eq{}
\eq{\ref{lem:triangles-fork-absorbs-anti-CNOT}}
\eq{\ref{lem:controlsep-swapped-wires}}
\end{align*}
\qed\end{proof}

\begin{proof}[Lemma \ref{lem:two-consecutive-controlsep-can-be-swapped}]
\def\fig{two-consecutive-controlsep-can-be-swapped}
\begin{align*}
&
\eq{\bt\\\ref{lem:hopf}}
\eq{\ref{lem:two-consecutive-controlsep-can-share-CNOT}}
\eq{\ref{lem:controlsep-surrounded-by-h}}\\&
\eq{\ref{lem:two-consecutive-controlsep-can-share-CNOT}}
\eq{\ref{lem:controlsep-surrounded-by-h}\\\ref{lem:two-consecutive-controlsep-can-share-CNOT}}\input{./figures/\fig/\fig_05.tikz}
\end{align*}
\qed\end{proof}

\begin{proof}[Lemma \ref{lem:red-state-on-controlsep}]
\def\fig{red-state-on-controlsep}
First:
\begin{align*}

\eq{}
\eq{\bo}
\eq{\ref{lem:red-state-on-triangle}\\\ref{lem:red-state-on-upside-down-triangle}}
\eq{\bo}
\end{align*}
Then:
\begin{align*}
\input{./figures/\fig/\fig_05.tikz}
\eq{\h\\\ref{lem:controlsep-surrounded-by-h}}\input{./figures/\fig/\fig_06.tikz}
\eq{}\input{./figures/\fig/\fig_07.tikz}
\eq{\h}\input{./figures/\fig/\fig_08.tikz}
\end{align*}
\qed\end{proof}

\begin{proof}[Lemma \ref{lem:anti-CNOT-on-controlsep}]
\def\fig{anti-CNOT-on-controlsep}
\begin{align*}

\eq{}
\eq{\bt}
\eq{\ref{lem:k1}\\\ref{lem:not-triangle-is-symmetrical}}
\eq{}
\end{align*}
\qed\end{proof}

\begin{proof}[Lemma \ref{lem:two-consecutive-controlsep-same-control}]
\def\fig{two-consecutive-controlsep-same-control}
\begin{align*}

\eq{\ref{lem:two-consecutive-controlsep-can-share-CNOT}}
\eq{\soo}
\eq{\ref{lem:hopf}}
\eq{}
\end{align*}
\qed\end{proof}

\begin{proof}[Lemma \ref{lem:two-consecutive-controlsep-anticontrol}]
\def\fig{two-consecutive-controlsep-anticontrol}
\begin{align*}

\eq{\ref{lem:two-consecutive-controlsep-can-share-CNOT}}
\eq{\ref{lem:hopf}}
\eq{}
\eq{\ref{lem:red-state-on-controlsep}}
\end{align*}
\qed\end{proof}

\begin{proof}[Lemma \ref{lem:controlsep-crosses-green-node}]
\def\fig{controlsep-crosses-green-node}
\begin{align*}
&
\eq{}
\eq{\ref{lem:triangles-fork-absorbs-anti-CNOT}}
\eq{\bt}\\&
\eq{\ref{lem:controlsep-swapped-wires}}
\eq{\ref{lem:triangles-fork-absorbs-anti-CNOT}}\input{./figures/\fig/\fig_05.tikz}
\eq{\stt\\\ref{lem:2-triangle-cycle}}\input{./figures/\fig/\fig_06.tikz}
\eq{}\input{./figures/\fig/\fig_07.tikz}
\end{align*}
\qed\end{proof}

\begin{proof}[Lemma \ref{lem:anti-controlsep}]
\def\fig{anti-controlsep}
\begin{align*}

\eq{}
\eq{\bt}
\eq{\ref{lem:looped-triangle}}
\eq{}
\end{align*}
\qed\end{proof}

\begin{proof}[Lemma \ref{lem:controlsep-and-control-alpha}]
\def\fig{controlsep-and-control-alpha}
\begin{align*}
&
\eq{\ref{lem:control-alpha-with-triangles}}
\eq{\ref{lem:controlsep-swapped-wires}}
\eq{\soo\\\ref{lem:controlsep-crosses-green-node}}\\&
\eq{\ref{lem:two-consecutive-controlsep-same-control}}
\eq{\ref{lem:controlsep-crosses-green-node}}\input{./figures/\fig/\fig_05.tikz}
\eq{\soo\\\ref{lem:controlsep-swapped-wires}}\input{./figures/\fig/\fig_06.tikz}
\end{align*}
\qed\end{proof}

\begin{proof}[Lemma \ref{lem:alpha-through-(anti)-controlsep}]
\def\fig{alpha-through-anti-controlsep}
\begin{align*}
&
\eq{\ref{lem:anti-controlsep}}
\eq{\ref{lem:controlsep-crosses-green-node}}
\eq{\ref{lem:controlsep-and-control-alpha}}\\&
\eq{\bt}
\eq{}\cdots
\eq{}\input{./figures/\fig/\fig_05.tikz}
\end{align*}
\qed\end{proof}

\begin{proof}[Lemma \ref{lem:gn-distrib-over-W}]
\def\fig{gn-distrib-over-W}
\begin{align*}
&
\eq{\ref{lem:2Id-is-C2-times-anti-C2}}
\eq{\soo\\\bo}
\eq{}\\&
\eq{\ref{lem:alpha-through-(anti)-controlsep}}
\eq{}\cdots
\eq{}\input{./figures/\fig/\fig_05.tikz}
\end{align*}
\qed\end{proof}

\begin{proof}[Lemma \ref{lem:cycle-triangle}]
Let $n$ be the number of triangles in the first two diagrams.\\
$\bullet$ $n=0$: The first equality is \ref{lem:h-loop}, the second is equivalent to the third, and easily derivable:
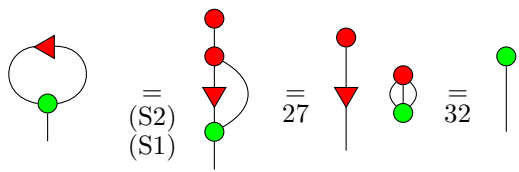
\def\fig{1-triangle-cycle-proof}
\begin{align*}

\eq{\stt\\\soo}
\eq{\ref{lem:looped-triangle}}
\eq{\ref{lem:red-state-on-upside-down-triangle}}
\end{align*}
$\bullet$ $n=1$: The first equality is \ref{lem:triangle-hadamard-parallel}, the second is \ref{lem:parallel-triangles} and the third is \ref{lem:2-triangle-cycle}.\\
$\bullet$ $n=2$: First,
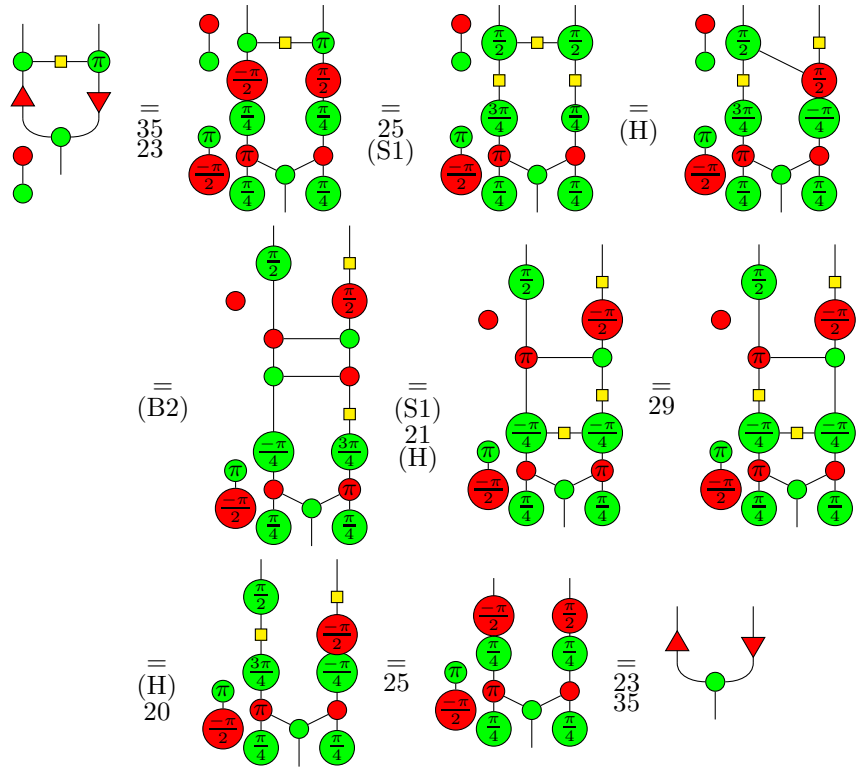
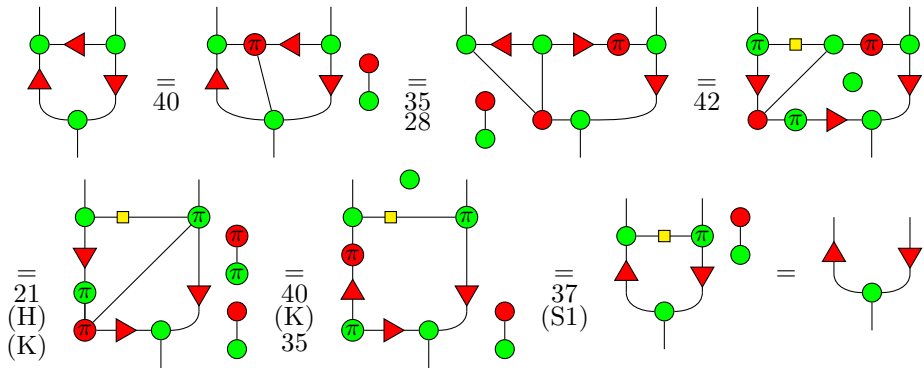
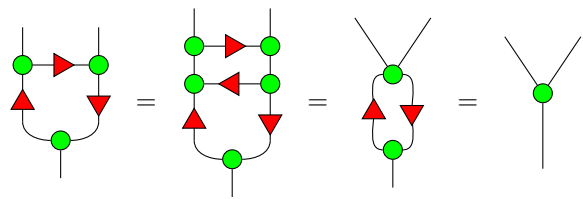
\def\fig{quasi-triangle-cycle-CZ-2}
\begin{align*}
&
\eq{\ref{lem:not-triangle-is-symmetrical}\\\ref{lem:supp-to-minus-pi_4}}
\eq{\ref{lem:euler-decomp-with-scalar}\\\soo}
\eq{\h}\\&
\eq{\bt}
\eq{\soo\\\ref{lem:k1}\\\h}\input{./figures/\fig/\fig_05.tikz}
\eq{\ref{lem:C1-original}}\input{./figures/\fig/\fig_06.tikz}\\&
\eq{\h\\\ref{lem:hopf}}\input{./figures/\fig/\fig_07.tikz}
\eq{\ref{lem:euler-decomp-with-scalar}}\input{./figures/\fig/\fig_08.tikz}
\eq{\ref{lem:supp-to-minus-pi_4}\\\ref{lem:not-triangle-is-symmetrical}}\input{./figures/\fig/\fig_09.tikz}
\end{align*}
Then:
\def\fig{quasi-triangle-cycle-2}
\begin{align*}

\eq{\ref{lem:triangles-fork-absorbs-anti-CNOT}}
\eq{\ref{lem:not-triangle-is-symmetrical}\\\ref{lem:black-dot-swappable-outputs}}
\eq{\ref{lem:old-bw}}\\
\eq{\ref{lem:k1}\\\h\\\kt}
\eq{\ref{lem:triangles-fork-absorbs-anti-CNOT}\\\kt\\\ref{lem:not-triangle-is-symmetrical}}\input{./figures/\fig/\fig_05.tikz}
\eq{\ref{lem:bw-triangle}\\\soo}\input{./figures/\fig/\fig_06.tikz}
\eq{}\input{./figures/\fig/\fig_07.tikz}
\end{align*}
Finally:
\def\fig{triangle-cycle-3}
\begin{align*}

\eq{}
\eq{}
\eq{}
\end{align*}
$\bullet$ $n$: Suppose we have the result for $n-1$ and $n=2$. Then:
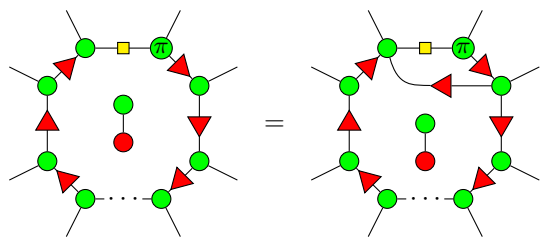
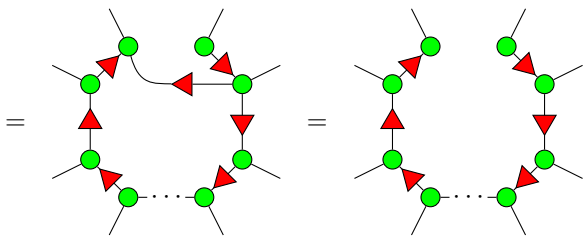
\def\fig{quasi-triangle-cycle-CZ-proof}
\begin{align*}
&
\eq{}\\&
\eq{}
\eq{}
\end{align*}
The same trick is used for the two other equalities.
\qed\end{proof}

\begin{proof}[Lemma \ref{lem:C2-distributed-through-W-2}]
first
\def\fig{C2-distributed-through-W-1}
\begin{align*}
&
\eq{\ref{lem:bw-triangle}}
\eq{\h}
\eq{\ref{lem:symmetric-diagram-with-triangle-hadamard}}
\eq{\ref{lem:not-triangle-is-symmetrical}\\\ref{lem:k1}\\\soo\\\stt}\\&
\eq{\soo\\\bo}\input{./figures/\fig/\fig_05.tikz}
\eq{\ref{lem:triangle-through-W}}\input{./figures/\fig/\fig_06.tikz}
\end{align*}
Then
\def\fig{C2-distributed-through-W-2}
\begin{align*}

\eq{\bt\\\ref{lem:2Id-is-C2-times-anti-C2}}
\eq{}
\end{align*}
Finally
\def\fig{C-half-distributed-through-W}
\begin{align*}

\eq{\ref{lem:2Id-is-C2-times-anti-C2}}
\eq{}
\eq{\ref{lem:2Id-is-C2-times-anti-C2}}
\end{align*}
\qed\end{proof}

\begin{proof}[Lemma \ref{lem:Calpha-through-W}]
\def\fig{Calpha-through-W}
\begin{align*}
&
\eq{\ref{lem:control-alpha-with-triangles}\\\bo}
\eq{\ref{lem:C2-distributed-through-W-2}\\\ref{lem:gn-distrib-over-W}}
\eq{\bt}\\&
\eq{\ref{lem:black-dot-swappable-outputs}\\\ref{lem:triangle-through-W}\\\soo}
\eq{\bt}\input{./figures/\fig/\fig_05.tikz}
\eq{\ref{lem:triangle-through-W}}\input{./figures/\fig/\fig_06.tikz}
\eq{\ref{lem:gn-distrib-over-W}\\\bo\\\soo}\input{./figures/\fig/\fig_07.tikz}\\&
\eq{\ref{lem:looped-triangle}}\input{./figures/\fig/\fig_08.tikz}
\eq{\ref{lem:triangle-through-W}\\\soo}\input{./figures/\fig/\fig_09.tikz}
\eq{\ref{lem:C2-distributed-through-W-2}}\input{./figures/\fig/\fig_10.tikz}
\eq{\ref{lem:symmetric-diagram-with-triangle-hadamard}\\\h}\input{./figures/\fig/\fig_11.tikz}\\&
\eq{\ref{lem:not-triangle-is-symmetrical}\\\ref{lem:k1}\\\h}\input{./figures/\fig/\fig_12.tikz}
\eq{\ref{lem:triangles-fork-absorbs-anti-CNOT}}\input{./figures/\fig/\fig_13.tikz}
\eq{\ref{lem:cycle-triangle}}\input{./figures/\fig/\fig_14.tikz}\\&
\eq{\soo\\\h\\\ref{lem:not-triangle-is-symmetrical}}\input{./figures/\fig/\fig_15.tikz}
\eq{\ref{lem:symmetric-diagram-with-triangle-hadamard}}\input{./figures/\fig/\fig_16.tikz}
\eq{\ref{lem:control-alpha-with-triangles}\\\bo}\input{./figures/\fig/\fig_17.tikz}
\end{align*}
\qed\end{proof}

\begin{proof}[Lemma \ref{lem:Cx-through-(anti)-controlsep}]
\def\fig{Cx-through-anti-controlsep}
\begin{align*}
&
\eq{\ref{lem:looped-controlsep}\\\ref{lem:controlsep-crosses-green-node}}
\eq{}\\&
\eq{\ref{lem:Calpha-through-W}}
\eq{}
\end{align*}
\qed\end{proof}

\begin{proof}[Lemma \ref{lem:other-forms-of-concatenation-gadget}]
\def\fig{other-forms-of-concatenation-gadget}
\begin{align*}

\eq{\ref{lem:controlsep-crosses-green-node}}
\eq{\ref{lem:anti-controlsep}}
\eq{\ref{lem:controlsep-swapped-wires}\\\ref{lem:k1}}
\end{align*}
\qed\end{proof}

\begin{proof}[Lemma \ref{lem:concatenation-gadget-NOT}]
\def\fig{concatenation-gadget-NOT}
\begin{align*}

\eq{\ref{lem:other-forms-of-concatenation-gadget}}
\eq{\ref{lem:k1}}
\eq{\ref{lem:other-forms-of-concatenation-gadget}}
\end{align*}
\qed\end{proof}

\begin{proof}[Lemma \ref{lem:concatenation-gadget-r-0}]
\def\fig{concatenation-gadget-r-0}
\begin{align*}

\eq{\soo\\\stt}
\eq{}
\eq{\ref{lem:cycle-triangle}\\\ref{lem:looped-triangle}}
\end{align*}
\qed\end{proof}

\begin{proof}[Lemma \ref{lem:concatenation-gadget-l-0}]
\def\fig{concatenation-gadget-l-0}
\begin{align*}

\eq{\ref{lem:concatenation-gadget-NOT}}
\eq{\ref{lem:concatenation-gadget-r-0}\\\soo\\\stt}
\end{align*}
\qed\end{proof}

\begin{proof}[Lemma \ref{lem:looped-controlsep}]
\def\fig{looped-controlsep}
\begin{align*}

\eq{\ref{lem:CNOT-under-controlsep}}
\eq{\soo\\\ref{lem:hopf}\\\stt}
\eq{\ref{lem:red-state-on-controlsep}}
\end{align*}
and
\begin{align*}

\eq{\ref{lem:swapped-CNOT-under-controlsep}\\\soo}\input{./figures/\fig/\fig_05.tikz}
\eq{\ref{lem:hopf}\\\stt}\input{./figures/\fig/\fig_06.tikz}
\eq{}\input{./figures/\fig/\fig_07.tikz}
\end{align*}
\qed\end{proof}

\begin{proof}[Lemma \ref{lem:CUP-on-(anti)-controlsep}]
\def\fig{CUP-on-anti-controlsep}
\begin{align*}
&
\eq{\ref{lem:controlsep-swapped-wires}}
\eq{\ref{lem:k1}\\\ref{lem:not-triangle-is-symmetrical}}
\eq{\ref{lem:control-alpha-with-triangles}}\\&
\eq{\bo}
\eq{\ref{lem:2Id-is-C2-times-anti-C2}\\\ref{lem:not-triangle-is-symmetrical}}\input{./figures/\fig/\fig_05.tikz}
\eq{\bo}\input{./figures/\fig/\fig_06.tikz}\\&
\eq{\ref{lem:k1}}\input{./figures/\fig/\fig_07.tikz}
\eq{\ref{lem:controlsep-swapped-wires}}\input{./figures/\fig/\fig_08.tikz}
\end{align*}
\qed\end{proof}

\begin{proof}[Lemma \ref{lem:gn-1-0-on-control-state-intermed-lemma}]
\def\fig{gn-1-0-on-control-state-intermed-lemma-bis}
\begin{align*}
&
\eq{\ref{lem:other-forms-of-concatenation-gadget}}\\&
\eq{\ref{lem:controlsep-crosses-green-node}\\\ref{lem:two-consecutive-controlsep-can-be-swapped}}
\eq{\ref{lem:anti-CNOT-on-controlsep}}\\&
\eq{\soo\\\ref{lem:controlsep-crosses-green-node}}
\eq{\ref{lem:CUP-on-(anti)-controlsep}}\input{./figures/\fig/\fig_05.tikz}\\&
\eq{\ref{lem:controlsep-crosses-green-node}\\\ref{lem:anti-CNOT-on-controlsep}}\input{./figures/\fig/\fig_06.tikz}
\eq{\ref{lem:controlsep-swapped-wires}}\input{./figures/\fig/\fig_07.tikz}\\&
\eq{\ref{lem:k1}\\\ref{lem:two-consecutive-controlsep-can-be-swapped}}\input{./figures/\fig/\fig_08.tikz}
\eq{\ref{lem:other-forms-of-concatenation-gadget}}\input{./figures/\fig/\fig_09.tikz}
\end{align*}
\qed\end{proof}

\begin{proof}[Lemma \ref{lem:(anti)-controlsep-under-addition}]
\def\fig{anti-controlsep-under-addition}
\begin{align*}
&
\eq{\ref{lem:2Id-is-C2-times-anti-C2}}
\eq{\bo}\\&
\eq{\ref{lem:C2-distributed-through-W-2}}
\eq{}\\&
\eq{\ref{lem:gn-1-0-on-control-state-intermed-lemma}}\input{./figures/\fig/\fig_05.tikz}
\eq{\bo\\\ref{lem:2Id-is-C2-times-anti-C2}}\input{./figures/\fig/\fig_06.tikz}
\end{align*}
\qed\end{proof}

\begin{proof}[Lemma \ref{lem:cos-alpha}]
\def\fig{cos-alpha}
\begin{align*}

\eq{\bo\\\soo}
\eq{\ref{lem:controlsep-and-control-alpha}}
\eq{\ref{lem:not-triangle-is-symmetrical}}
\eq{\ref{lem:control-alpha-with-triangles}}
\eq{\soo\\\stt}\input{./figures/\fig/\fig_05.tikz}
\end{align*}
\qed\end{proof}

\begin{proof}[Lemma \ref{lem:(anti)-controlsep-triangle}]
\def\fig{anti-controlsep-triangle}
\begin{align*}

\eq{\stt\\\soo\\\bo}
\eq{\ref{lem:(anti)-controlsep-under-addition}}
\eq{\ref{lem:red-state-on-controlsep}\\\soo\\\bo\\\stt}
\end{align*}
\qed\end{proof}

\begin{proof}[Lemma \ref{lem:commutation-base-component}]
First, if $a=1=b$:
\def\fig{base-building-component-commutation-aux}
\begin{align*}

\eq{\soo\\\bo}
\eq{\ref{lem:gn-distrib-over-W}}
\eq{\soo\\\ref{lem:triangle-through-W}}
\eq{\ref{lem:gn-distrib-over-W}}
\eq{\bo\\\soo}\input{./figures/\fig/\fig_05.tikz}
\end{align*}
Then:
\def\fig{base-building-component-commutation}
\begin{align*}

\eq{\stt\\\soo}
\eq{}
\eq{\soo\\\stt}
\end{align*}
\qed\end{proof}

\begin{proof}[Lemma \ref{lem:polynomial-through-triangle}]
\def\fig{polynomial-through-triangle}
\begin{align*}

\eq{\stt\\\soo\\\ref{lem:k1}\\\ref{lem:not-triangle-is-symmetrical}}
\eq{\ref{lem:anti-CNOT-on-controlsep}\\\ref{lem:red-state-on-controlsep}}
\eq{}\\
\eq{}
\eq{\ref{lem:gn-distrib-over-W}\\\ref{lem:triangle-through-W}}\input{./figures/\fig/\fig_05.tikz}\\
\eq{\ref{lem:Cx-through-(anti)-controlsep}\\\ref{lem:(anti)-controlsep-triangle}\\\kt\\\ref{lem:not-triangle-is-symmetrical}}\input{./figures/\fig/\fig_06.tikz}
\eq{}\cdots\eq{\ref{lem:commutation-base-component}}\input{./figures/\fig/\fig_07.tikz}
\eq{}\cdots\\
\eq{}\input{./figures/\fig/\fig_08.tikz}
\eq{\soo\\\stt\\\ref{lem:cycle-triangle}\\\ref{lem:concatenation-gadget-l-0}}\input{./figures/\fig/\fig_09.tikz}
\eq{\ref{lem:not-triangle-is-symmetrical}\\\ref{lem:k1}}\input{./figures/\fig/\fig_10.tikz}
\eq{\ref{lem:triangles-fork-absorbs-anti-CNOT}}\input{./figures/\fig/\fig_11.tikz}
\eq{\ref{lem:cycle-triangle}}\input{./figures/\fig/\fig_12.tikz}
\end{align*}
\qed\end{proof}

\begin{proof}[Lemma \ref{lem:red-state-on-polynomial}]
\def\fig{red-state-on-polynomial}
First if $P=0$:
\begin{align*}

\eq{}
\end{align*}
Then, if $P(X)=P'(X)+(-1)^baX^k$:
\begin{align*}

\eq{\soo\\\bo\\\ref{lem:bicolor-0-alpha}\\\ref{lem:red-state-on-triangle}}
\eq{}
\eq{}\input{./figures/\fig/\fig_05.tikz}
\end{align*}
\qed\end{proof}

\begin{proof}[Lemma \ref{lem:polynomial-times-not-polynomial}]
\def\fig{polynomial-times-not-polynomial}
\begin{align*}

\eq{\stt\\\soo\\\ref{lem:pi-red-state-on-triangle}\\\bo\\\ref{lem:k1}}
\eq{\ref{lem:polynomial-through-triangle}}
\eq{\ref{lem:k1}\\\bo\\\ref{lem:pi-red-state-on-triangle}}
\end{align*}
\qed\end{proof}

\begin{proof}[Lemma \ref{lem:add-axiom-triangle}]
\def\fig{add-axiom-triangle}
\begin{align*}
\zx^{\textnormal A}\vdash~~&
\eq{\soo\\\bo}
\eq{\add}
\eq{\bt}\\&
\eq{\bt}
\eq{}\input{./figures/\fig/\fig_05.tikz}
\end{align*}
\qed\end{proof}

\begin{proof}[Lemma \ref{lem:prod-cos}]
\def\fig{prod-cos}
\begin{align*}
&\hspace*{-1em}
\eq{\ref{lem:cos-alpha}}
\eq{\ref{lem:Calpha-through-W}}
\eq{\ref{lem:cos-alpha}}\\&
\eq{\ref{lem:gn-distrib-over-W}\\\ref{lem:C2-distributed-through-W-2}}
\eq{\ref{lem:black-dot-swappable-outputs}\\\ref{lem:triangle-through-W}}\input{./figures/\fig/\fig_05.tikz}
\eq{\ref{lem:C2-distributed-through-W-2}}\input{./figures/\fig/\fig_06.tikz}\\&
\eq{\ref{lem:cos-alpha}}\input{./figures/\fig/\fig_07.tikz}
\eq{\ref{lem:add-axiom-triangle}}\input{./figures/\fig/\fig_08.tikz}
\eq{\ref{lem:C2-distributed-through-W-2}}\input{./figures/\fig/\fig_09.tikz}
\end{align*}
with $\cos(\gamma) = \frac{1}{2}(\cos(\alpha-\beta)+ \cos(\alpha+\beta)) = \cos(\alpha)\cos(\beta)$.
\qed\end{proof}

\begin{proof}[Lemmas \ref{lem:add-bis}, \ref{lem:ctrl-power-2} and Corollary \ref{cor:arccos-2-power-n}]
These lemmas were already proven in \cite{JPV-universal}, but now the proofs are constructive since \ref{lem:prod-cos} has a constructive proof.
\qed\end{proof}


\subsection{Monoid and Bialgebra}

\begin{proof}[Proposition \ref{prop:monoid}]~
\begin{itemize}
\item Commutativity: Lemma \ref{lem:controlsep-swapped-wires}
\item Unit: \soo, Lemma \ref{lem:2-triangle-cycle}
\item Associativity: Lemma \ref{lem:two-consecutive-controlsep-can-be-swapped}
\end{itemize}
\qed\end{proof}

\begin{proof}[Proposition \ref{prop:bialgebra}]~
\begin{itemize}
\item Associativity and unit: \soo and \stt
\item Coassociativity and couint: Proposition \ref{prop:monoid}
\item Unit and counit: Lemmas \ref{lem:bicolor-0-alpha} and \ref{lem:inverse}
\item Multiplication and counit: Lemma \ref{lem:k1} and \bo
\item Comultiplication and unit: Lemma \ref{lem:red-state-on-controlsep}
\item Multiplication and comultiplication: Lemmas \ref{lem:k1} and \ref{lem:controlsep-crosses-green-node}
\end{itemize}
\qed\end{proof}

\subsection{Preliminary Results for Completeness}

\begin{lemma}
\label{lem:defCNF}
\[\def\fig{control-psi-state-0}\eq{}\]
\end{lemma}

\begin{proof}

First, let $\ket{\psi_0}$ and $\ket{\psi_1} \in\mathcal{R}^{2^n}$ such that $\ket{\psi}=\ket{0}\ket{\psi_0}+\ket{1}\ket{\psi_1}$. Then:
\def\fig{control-psi-state-0-proof}
\begin{align*}
\zxc\vdash~~&
\eq{}
\eq{\bo\\\soo\\\stt}\\&
\eq{\ref{lem:red-state-on-controlsep}}
\eq{\bo}
\eq{\textit{Ind}\\\soo}\input{./figures/\fig/\fig_05.tikz}
\end{align*}
It then remains to prove the result for the base cases $\Lambda x$. Any $x$ can be decomposed as a sum of $e^{i\alpha}$ where $\alpha$s are in the fragment. Then:
\def\fig{control-exp-i-alpha-0}
\[\zxc\vdash~~
\eq{}
\eq{\ref{lem:bicolor-0-alpha}}
\eq{\ref{lem:inverse}}\]
and:
\def\fig{sum-of-control-scalars-0}
\begin{align*}
\zxc\vdash~~&
\eq{}
\eq{\bo\\\ref{lem:red-state-on-triangle}}\\&
\eq{\ref{lem:inverse}\\\bo}
\eq{}
\end{align*}
\qed\end{proof}

\begin{lemma}
\label{lem:Lambda-x-distrib-over-W}
\[
\input{./figures/distrib-of-control-scalar.tikz}
\]
\end{lemma}

\begin{proof}
Let $x$ be in $\mathcal{R}_G$ for some fragment $G$. Then there exist $p$, $\vec \alpha=(\alpha_k)_k$ and $\vec P=(P_k)_k$ such that $x=\frac{1}{2^p}\sum\limits_k P(e^{i\alpha_k})$. The conditions for Theorem \ref{them:completeness} imply that:
\begin{align*}
\zxc\vdash~~
\def\fig{triangle-x-generic-for-distrib}&
\eq{}
\input{./figures/control-x-general.tikz}

\end{align*}
Then:
\def\fig{polynomial-distrib-over-W}
\begin{align*}
\zxc\vdash~~&
\eq{\ref{lem:C2-distributed-through-W-2}}\\&
\eq{\bt}
\eq{\ref{lem:polynomial-through-triangle}}
\end{align*}
\qed\end{proof}

\begin{lemma}
\label{lem:deducible-control-scalar}
With \zxc a set of axioms that verifies the conditions in Theorem \ref{them:completeness}:
\begin{align*}
\zxc\vdash ~~&
\def\fig{control-0}\left(\eq{}\right),
\def\fig{control-2}\left(\eq{}\right),
\def\fig{control-1_2}\left(\eq{}\right),\\&
\def\fig{control-1_sqrt-2}\left(\eq{}\right),
\left(
\input{./figures/control-minus-1_sqrt-2.tikz}
\right),
\def\fig{control-2-1}\left(\eq{}\right),\\
&\def\fig{control-1_sqrt-2-1}\left(\eq{}\input{./figures/\fig/\fig_05.tikz}\right)
\end{align*}
\end{lemma}

\begin{proof}[Lemma \ref{lem:deducible-control-scalar}]
Since $\zxc\vdash \zx_{\frac{\pi}{4}}$:
\def\fig{control-0}
\begin{align*}
\zxc\vdash~~
\eq{}
\eq{}
\eq{\ref{lem:k1}\\\bo\\\soo}
\eq{\ref{lem:pi-green-state-on-upside-down-triangle}}
\end{align*}
\def\fig{control-2}
\begin{align*}
\zxc\vdash~~
\eq{}
\eq{}
\eq{\bo\\\soo\\\stt}
\end{align*}
\def\fig{control-1_2}
\begin{align*}
\zxc\vdash~~
\eq{}
\eq{}
\eq{\ref{lem:2Id-is-C2-times-anti-C2}}
\end{align*}
\def\fig{control-1_sqrt-2}
\begin{align*}
\zxc\vdash
\eq{\ref{lem:cos-alpha}}
~\eq[]{}
\eq[~]{}
\eq{}
\end{align*}
\def\fig{control-2-1}
\begin{align*}
\zxc\vdash~~
\eq{}
\eq{\ref{lem:pi-red-state-on-triangle}}
\eq{\ref{lem:inverse}\\\soo}
\end{align*}
\def\fig{control-1_sqrt-2-1}
\begin{align*}
\zxc\vdash~~
\eq[~]{}
\eq[~]{\soo\\\kt}
\eq[~]{\ref{lem:green-state-pi_2-is-red-state-minus-pi_2}\\\ref{lem:bicolor-0-alpha}}
\eq[~]{\ref{lem:multiplying-global-phases}\\\ref{lem:bicolor-0-alpha}}
\eq[~]{\ref{lem:inverse}}\input{./figures/\fig/\fig_05.tikz}
\end{align*}
\qed\end{proof}

\subsection{Necessary Propositions for Theorem \ref{them:completeness}}

In this section, we consider a set of diagrams $S$ such that the map $\eta:S\to \mathcal{R}_G = D\mapsto \interp{D}\ket{1}$ is bijective, and a set of axioms $\zxc$ that meets the conditions of Theorem \ref{them:completeness}.

Let us try to derive the results on the composition of diagrams in \snf, and see what it requires from diagrams in \scnf. Let $D_1$ and $D_2$ be two diagrams in \snf. First for the spacial composition:

\def\fig{tensor-of-normal-forms}
\begin{align*}
&
\eq{}\\&
\eq{\bo\\\ref{lem:k1}}
\eq{ii)}\\&
\eq{i)}
\end{align*}

This proof requires that, $i)$ a diagram in \scnf with a permutation on the output wires can be set in \scnf:

\begin{proposition}[Permutation]
\label{prop:swap-pnf}
For any $\ket{\psi}\in\mathcal{R}^{2^n}$, and any permutation $\sigma$ on $n$ wires:
\[\zx_G\vdash~~
\input{./figures/permutation-on-outputs.tikz}
\]
\end{proposition}

$ii)$ that two diagrams in \scnf joint by their control wire by a green node can be set \scnf:

\begin{proposition}[Tensor Product]
\label{prop:tensor-pnf}
For any $(\ket{\psi_0},\ket{\psi_1})\in\mathcal{R}^{2^n}\times\mathcal{R}^{2^m}$:
\[\zxc\vdash~~
\input{./figures/tensor-product-of-normal-forms-prop.tikz}
\]
\end{proposition}

Then, for the sequential composition:

\def\fig{sequential-composition-of-normal-forms}
\begin{align*}
&
\eq{}
\eq{\bo\\\ref{lem:k1}}\\&
\eq{ii)}
\eq{iii)}
\end{align*}

Here, we need again Proposition \ref{prop:tensor-pnf}, but we also seem to need that a diagram in \scnf with a cup $\epsilon$ applied on two of its outputs can be set in \scnf. It is actually true only if controlled by $\ket{1}$:

\begin{proposition}[Trace]
\label{cor:epsilon-pnf}
For any diagram $D:0\to n+1$:
\[\zxc\vdash~~
\input{./figures/cup-on-control-state-ket-1.tikz}
\]
\end{proposition}

To do so, we decompose the cup as: 
\input{./figures/cup-as-RZ.tikz}
, so that we can prove:

\begin{proposition}[$R_Z^{(2,1)}$]
\label{prop:Z-2-1-pnf}
For any $D:0\to n+2$:
\[\zxc\vdash~~
\input{./figures/gn-2-1-on-control-state-prop.tikz}
\]
\end{proposition}

\begin{proposition}[$R_Z^{(1,0)}$]
\label{prop:Z-1-0-pnf}
For any diagram $D:0\to n+1$:
\[\zxc\vdash~~
\input{./figures/gn-1-0-on-control-state-prop.tikz}
\]
\end{proposition}

Then, Lemmas \ref{scnf:tensor} and \ref{scnf:composition} derive from Propositions \ref{prop:swap-pnf}, \ref{prop:tensor-pnf} and \ref{cor:epsilon-pnf}.

\begin{proof}[Proposition \ref{prop:swap-pnf}]
Any permutation can be decomposed in a sequence of adjacent transpositions, which in ZX translates as swaps $\sigma$. If $\ket{\psi}$ is a state on $0$ or $1$ qubit, the only permutation allowed is the identity. Otherwise, let $\ket{\psi}=\ket{0}\ket{\psi_0}+\ket{1}\ket{\psi_1}=\ket{00}\ket{\psi_{00}}+\ket{01}\ket{\psi_{01}}+\ket{10}\ket{\psi_{10}}+\ket{11}\ket{\psi_{11}}$. If the first wire is not affected by the swap:
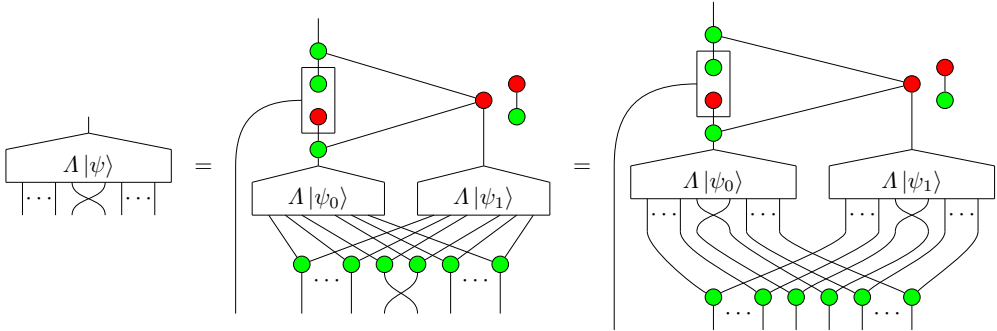
\def\fig{swap-on-control-state-easy}
\begin{align*}
\scalebox{0.87}{$
\eq{}
\eq{}$}
\end{align*}
which can be set in normal form by induction. If a swap occurs on the two first outputs:
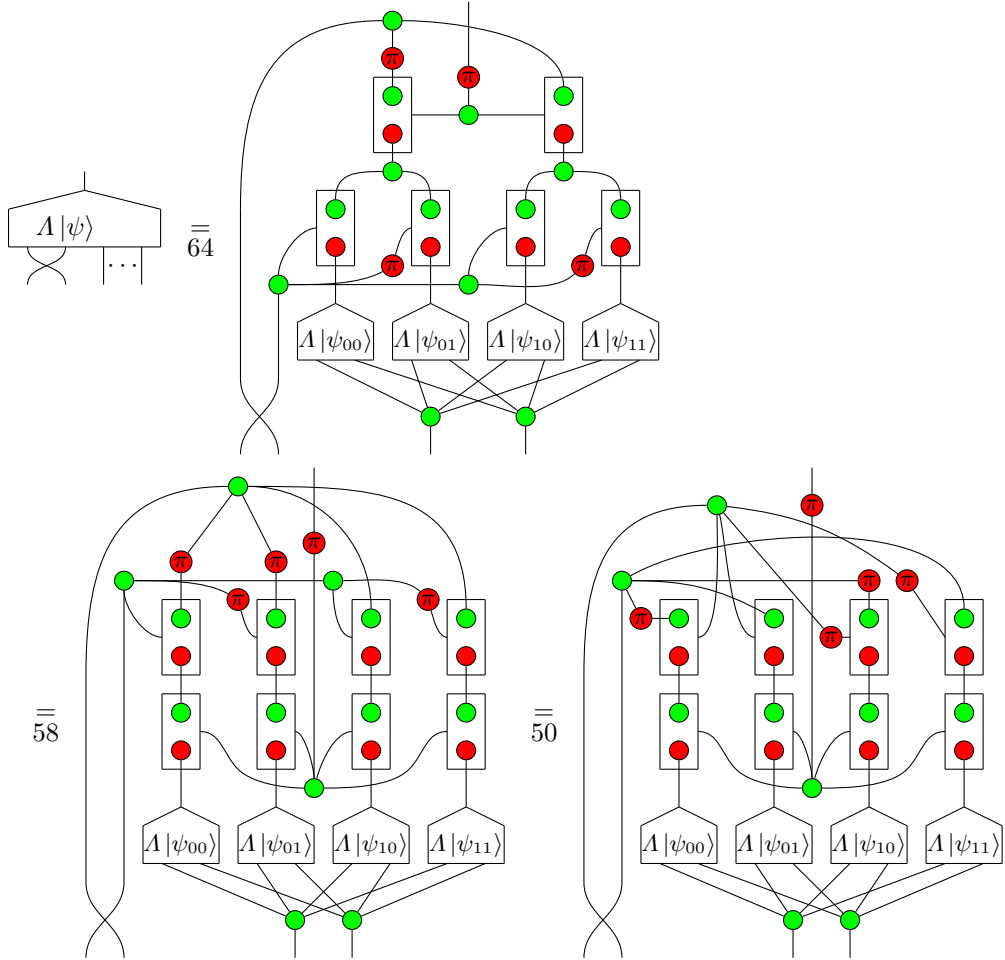
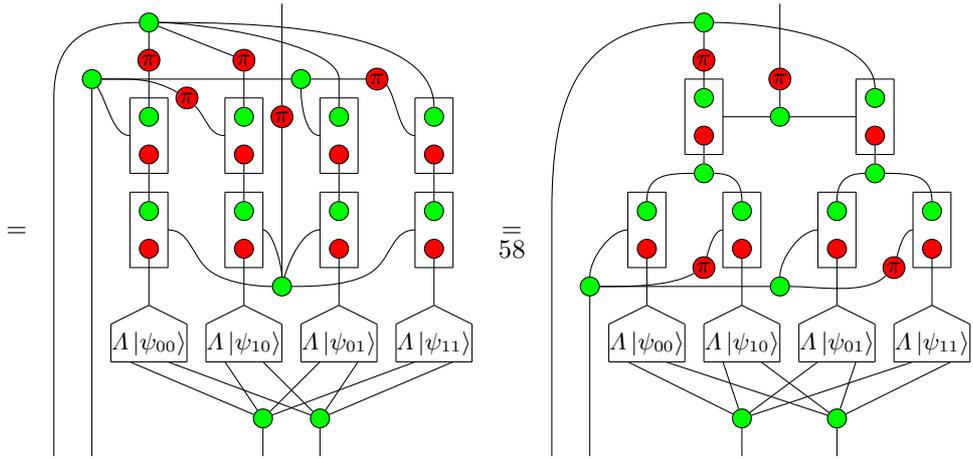
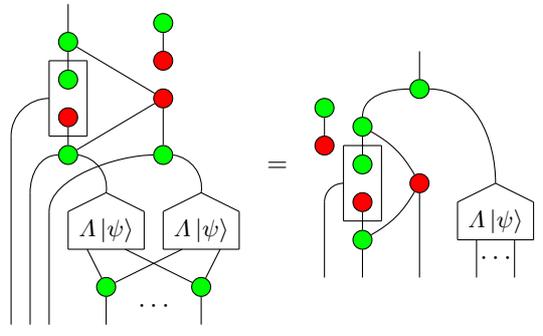
\def\fig{swap-on-control-state}
\begin{align*}
&
\eq{\ref{lem:other-forms-of-concatenation-gadget}}\\&
\eq{\ref{lem:controlsep-crosses-green-node}}
\eq{\ref{lem:controlsep-swapped-wires}}\\&
\eq{}
\eq{\ref{lem:controlsep-crosses-green-node}}\input{./figures/\fig/\fig_05.tikz}
\end{align*}
\qed\end{proof}

\begin{lemma}
\label{lem:control-and-anti-control-state}
\def\fig{control-and-anti-control-state}
\[\eq{}\input{./figures/\fig/\fig_08.tikz}\]
\end{lemma}

\begin{proof}[Lemma \ref{lem:control-and-anti-control-state}]
By induction on the number $n$ of outputs of $\ket{\psi}$:\\
$\bullet$ $n=0$: 
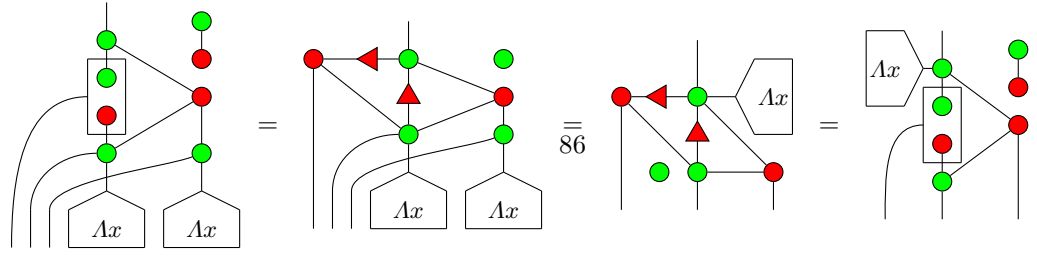
\def\fig{control-and-anti-control-state-easy}
\begin{align*}

\!\!\eq{}
\!\!\eq{\ref{lem:Lambda-x-distrib-over-W}}
\eq{}
\end{align*}
$\bullet$ $n\geq1$: In this case, let $\ket{\psi}=\ket{0}\ket{\psi_0}+\ket{1}\ket{\psi_1}$, and
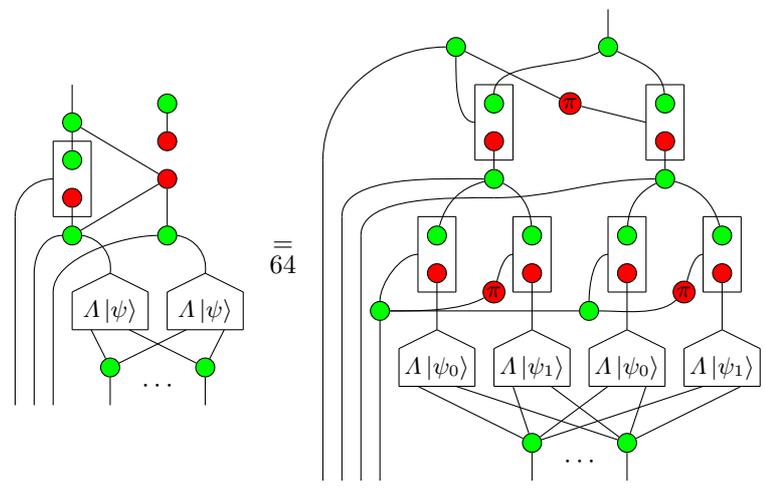
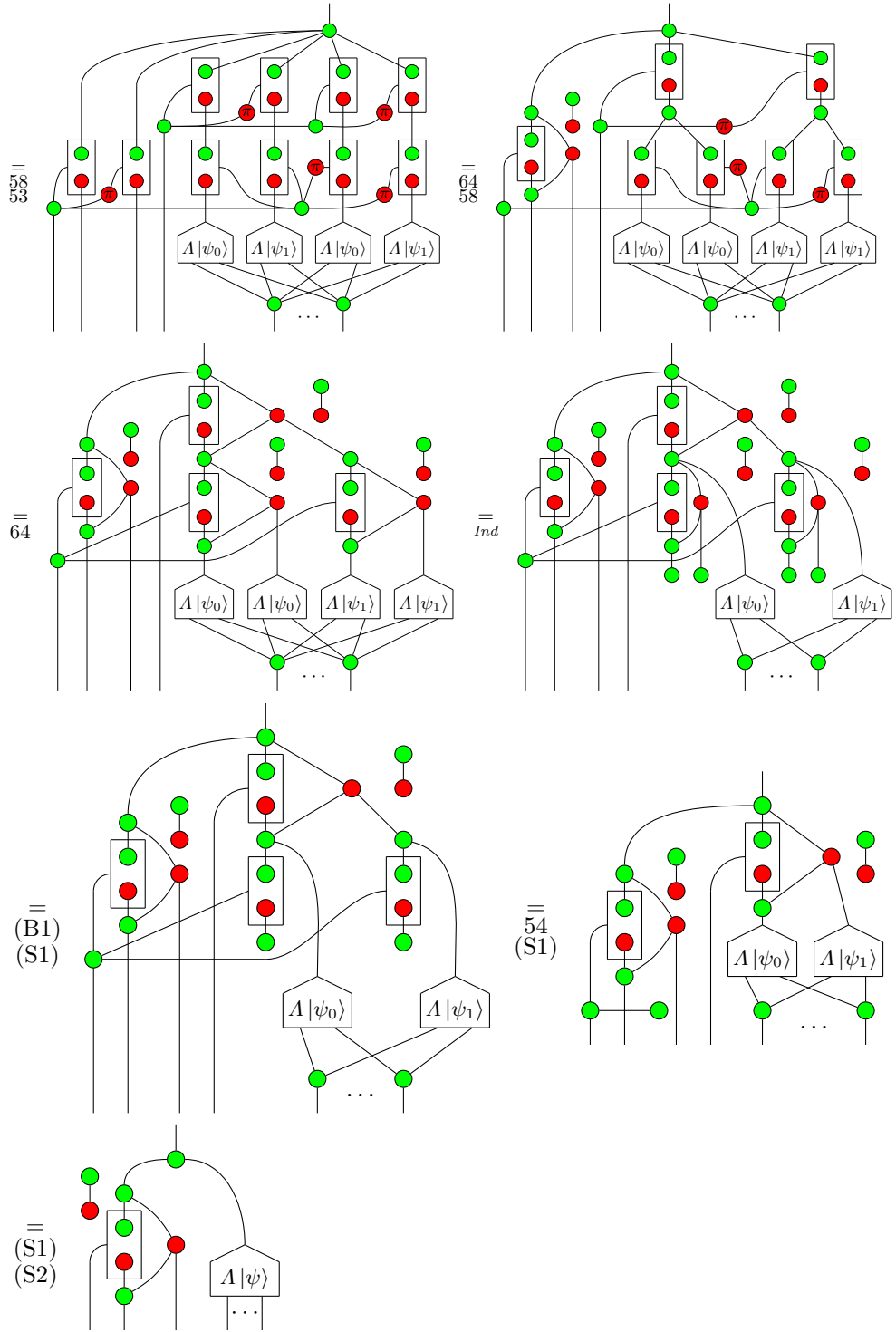
\def\fig{control-and-anti-control-state}
\begin{align*}
&
\eq{\ref{lem:other-forms-of-concatenation-gadget}}\\&
\scalebox{0.8}{$\eq{\ref{lem:controlsep-crosses-green-node}\\\ref{lem:two-consecutive-controlsep-can-be-swapped}}
\eq{\ref{lem:other-forms-of-concatenation-gadget}\\\ref{lem:controlsep-crosses-green-node}}$}\\&
\scalebox{0.85}{$\eq{\ref{lem:other-forms-of-concatenation-gadget}}
\eq{\textit{Ind}}\input{./figures/\fig/\fig_05.tikz}$}\\&
\eq{\bo\\\soo}\input{./figures/\fig/\fig_06.tikz}
\eq{\ref{lem:red-state-on-controlsep}\\\soo}\input{./figures/\fig/\fig_07.tikz}\\&
\eq{\soo\\\stt}\input{./figures/\fig/\fig_08.tikz}
\end{align*}
\qed\end{proof}

\begin{proof}[Proposition \ref{prop:tensor-pnf}]
By induction on the number of outputs of $\ket{\psi_0}$ and $\ket{\psi_1}$:\\
$\bullet$ If both states are scalars, this case is handled by the condition in Theorem \ref{them:completeness}.\\
$\bullet$ If one of the two states has at least one output -- say $\ket{\psi_0}=\ket{0}\ket{\psi_{00}}+\ket{1}\ket{\psi_{01}}$:
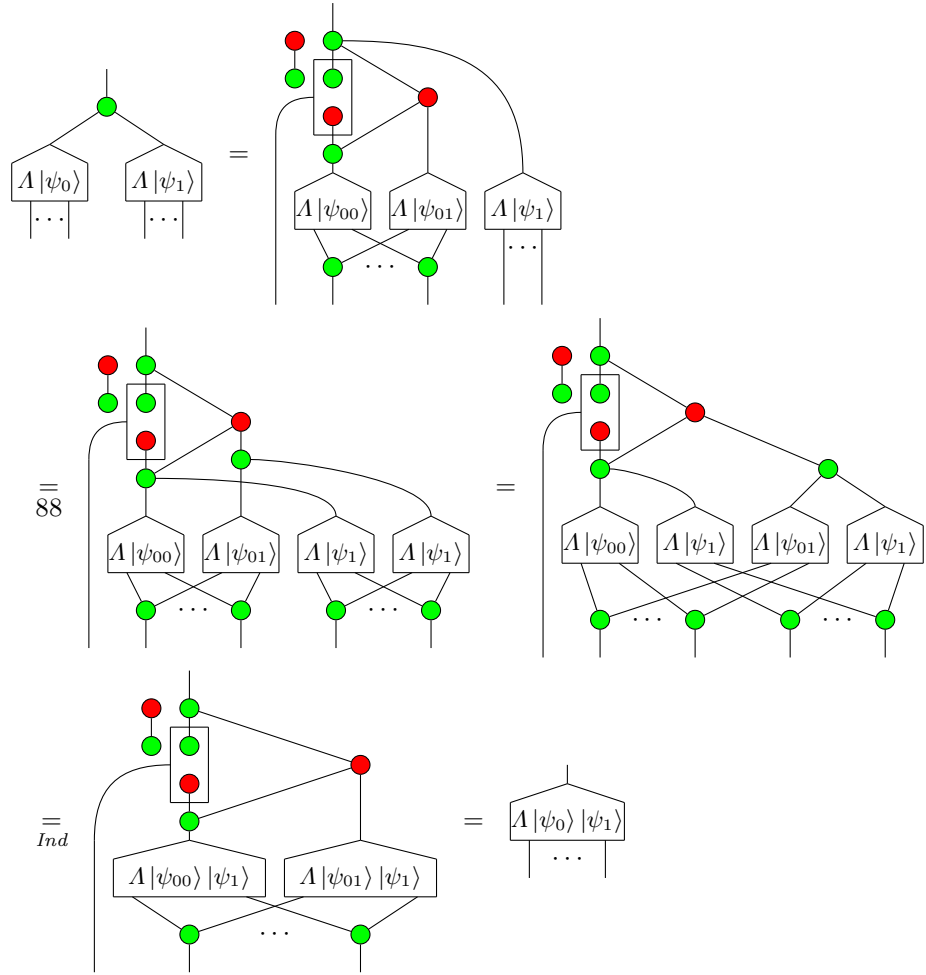
\def\fig{tensor-product-of-normal-forms-easy}
\begin{align*}
&
\eq{}\\&
\eq{\ref{lem:control-and-anti-control-state}}
\eq{}\\&
\eq{\textit{Ind}}
\eq{}\input{./figures/\fig/\fig_05.tikz}
\end{align*}
\qed\end{proof}

\begin{proof}[Proposition \ref{prop:Z-2-1-pnf}]
By induction on the number $n$ of outputs of $\ket{\psi}$.\\
$\bullet$ $n=2$: First notice:
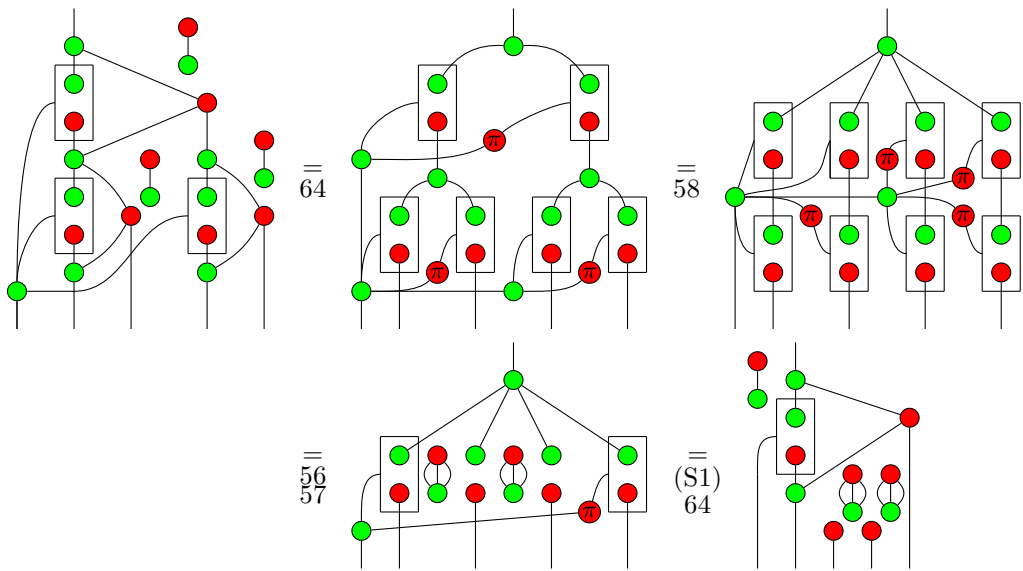
\def\fig{gn-2-1-on-control-state-base}
\begin{align*}
&
\eq{\ref{lem:other-forms-of-concatenation-gadget}}
\eq{\ref{lem:controlsep-crosses-green-node}}\\&
\eq{\ref{lem:two-consecutive-controlsep-same-control}\\\ref{lem:two-consecutive-controlsep-anticontrol}}
\eq{\soo\\\ref{lem:other-forms-of-concatenation-gadget}}
\end{align*}
Then, if $\ket\psi = a\ket{00}+b\ket{01}+c\ket{10}+\ket{11}$:
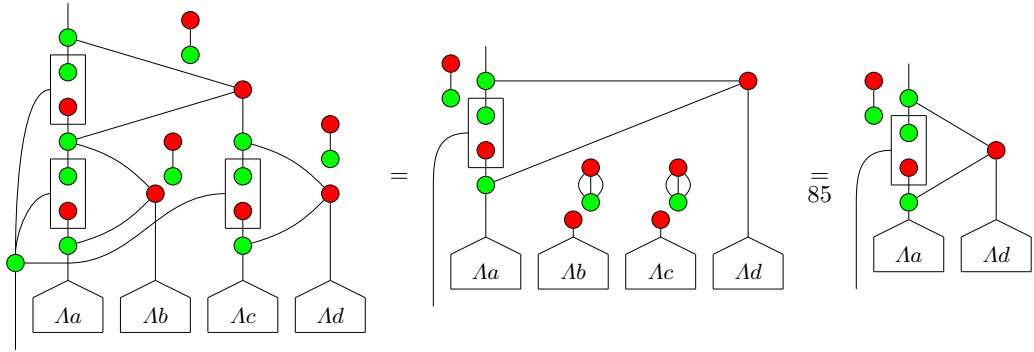
\def\fig{gn-2-1-on-control-state-base-2}
\begin{align*}
\scalebox{0.92}{$
\eq{}
\eq{\ref{lem:defCNF}}$}
\end{align*}
which is in normal form.\\
$\bullet$ $n\geq3$: Using Proposition \ref{prop:swap-pnf}, we can impose 
\input{./figures/gn-2-1.tikz}
 to be applied on the two last wires. Then:
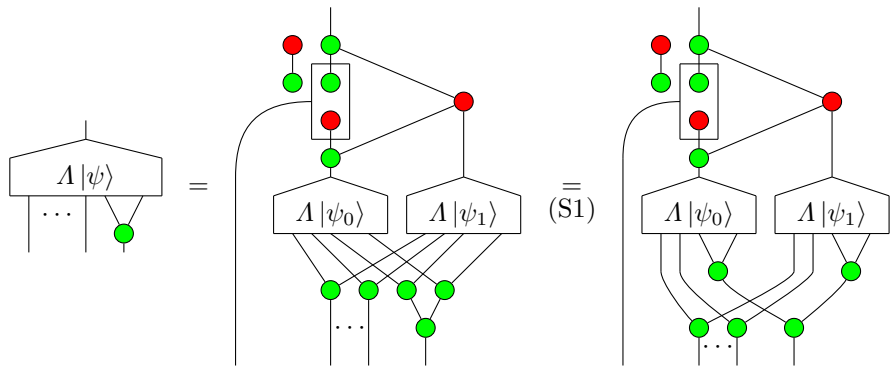
\def\fig{gn-2-1-on-control-state-easy}
\begin{align*}

\eq{}
\eq{\soo}
\end{align*}
\qed\end{proof}

\begin{proof}[Proposition \ref{prop:Z-1-0-pnf}]
By induction of the number $n$ of wires of $\ket{\psi}$:\\
$\bullet$ $n=1$: Let $\ket\psi = a\ket0 +b\ket1$. Then:
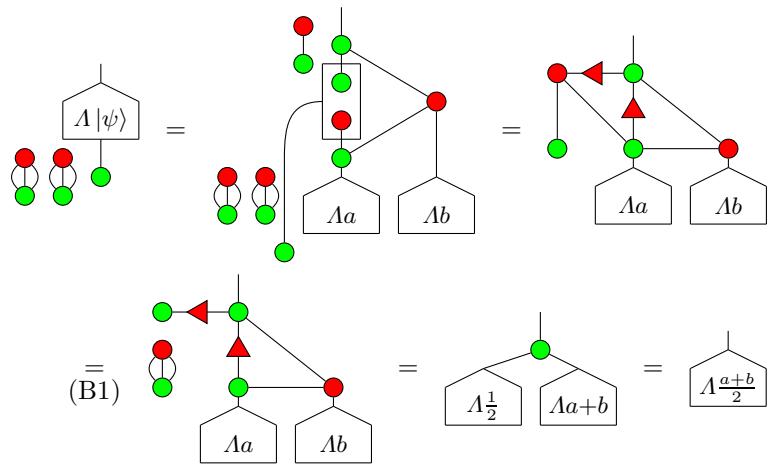
\def\fig{gn-1-0-on-control-state-base}
\begin{align*}

\eq{}
\eq{}\\
\eq{\bo}
\eq{}
\eq{}\input{./figures/\fig/\fig_05.tikz}
\end{align*}
$\bullet$ $n\geq2$: First, using Proposition \ref{prop:swap-pnf} if needs be,
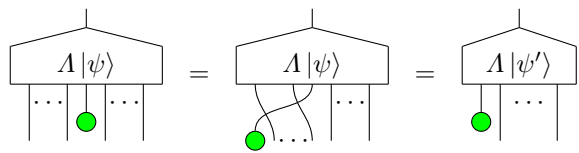
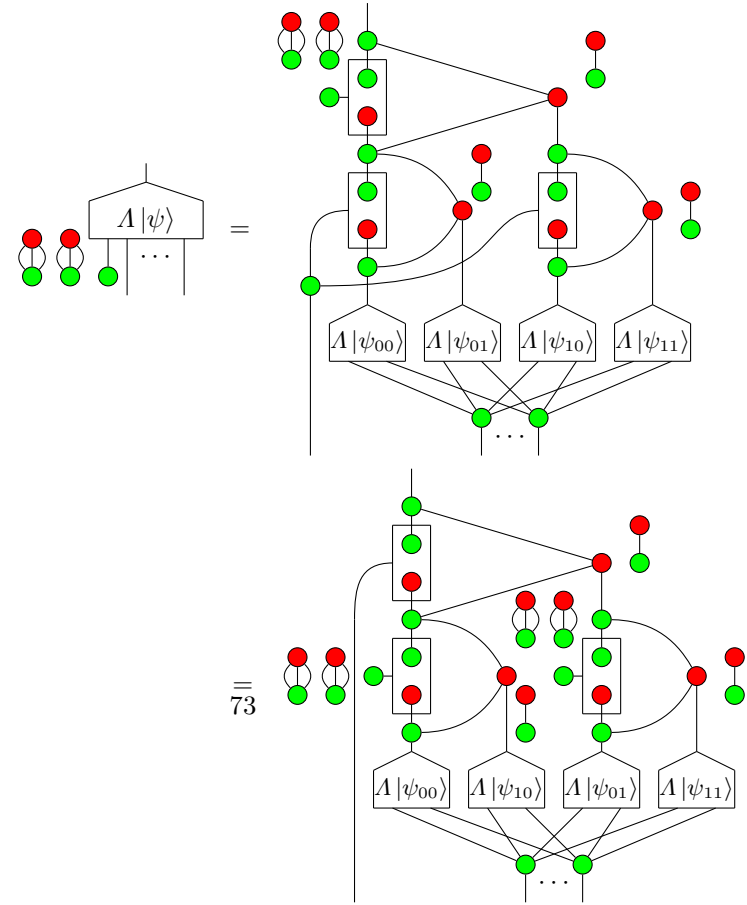
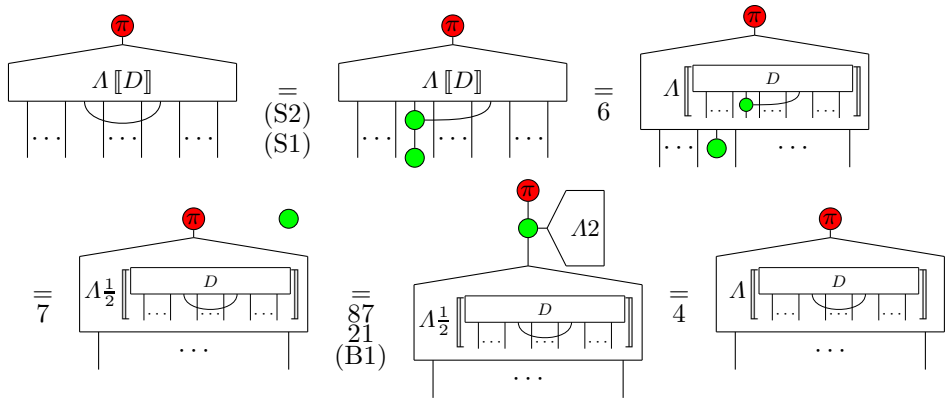
\def\fig{gn-1-0-on-control-state-easy}
\begin{align*}

\eq{}
\eq{}
\end{align*}
then,
\def\fig{gn-1-0-on-control-state}
\begin{align*}
&
\eq{}\\&
\eq{\ref{lem:gn-1-0-on-control-state-intermed-lemma}}
\end{align*}
\qed\end{proof}

\begin{proof}[Proposition \ref{cor:epsilon-pnf}]
\def\fig{cup-on-control-state-ket-1-proof}
\begin{align*}
&
\eq{\stt\\\soo}
\eq{\ref{prop:Z-2-1-pnf}}\\&
\eq{\ref{prop:Z-1-0-pnf}}
\eq{\ref{lem:deducible-control-scalar}\\\ref{lem:k1}\\\bo}
\eq{\ref{prop:tensor-pnf}}\input{./figures/\fig/\fig_05.tikz}
\end{align*}
\qed\end{proof}

\begin{proof}[Lemma \ref{scnf:generators}]
We will prove the result for states, for the three-legged green dot, the Hada\-mard node and the empty diagram. All the other generators can be built from them and the Propositions \ref{prop:swap-pnf}, \ref{prop:tensor-pnf}, \ref{prop:Z-2-1-pnf},\ref{cor:epsilon-pnf} and \ref{prop:Z-1-0-pnf}:
First, notice that:
\begin{align*}
\def\fig{lambda-ket-0}
\zxc\vdash~~
\eq{}
\eq{\ref{lem:concatenation-gadget-r-0}\\\soo\\\stt}
\qquad\text{and}\qquad
\def\fig{lambda-ket-1}

\eq{}
\eq{\ref{lem:concatenation-gadget-l-0}\\\soo\\\stt}
\end{align*}
Then:
\def\fig{gn-0-3-to-normal-form-bis}
\begin{align*}
&\zxc\vdash~~
\eq{\ref{lem:cycle-triangle}}
\eq{\stt\\\soo\\\ref{lem:k1}}\\&
\eq{}
\eq{\ref{prop:tensor-pnf}}
\eq{}\input{./figures/\fig/\fig_05.tikz}
\end{align*}
and:
\def\fig{Hadamard-to-normal-form-1}
\begin{align*}
\zxc\vdash~~&
\eq{\eu}
\eq{\ref{lem:control-alpha-with-triangles}}
\eq{\ref{lem:not-triangle-is-symmetrical}\\\stt\\\soo}\\&
\eq{\ref{lem:red-state-on-controlsep}\\\ref{lem:looped-controlsep}\\\bo}
\eq[]{\bo\\\ref{lem:k1}}~~\input{./figures/\fig/\fig_05.tikz}\\&
\eq{}\input{./figures/\fig/\fig_06.tikz}
\eq{\ref{lem:deducible-control-scalar}\\\ref{lem:k1}\\\bo}\input{./figures/\fig/\fig_07.tikz}
\eq{\ref{prop:tensor-pnf}}\input{./figures/\fig/\fig_08.tikz}
\end{align*}
and:
\def\fig{empty-diagram-to-pnf}
\begin{align*}
\zxc\vdash~~
\eq{\ref{lem:inverse}\\\ref{lem:bicolor-0-alpha}}
\eq{}
\end{align*}
Then:
\def\fig{Z-0-1-to-normal-form}
\begin{align*}
\zxc\vdash~~
\eq{\soo\\\stt}
\eq{}
\eq{\ref{cor:epsilon-pnf}}
\end{align*}
\def\fig{cap-to-normal-form}
\begin{align*}
\zxc\vdash~~&
\eq{\stt\\\soo}
\eq{}
\eq{\ref{prop:Z-1-0-pnf}}\\&
\eq{\ref{lem:deducible-control-scalar}\\\ref{lem:k1}\\\bo}
\eq{\ref{prop:tensor-pnf}}\input{./figures/\fig/\fig_05.tikz}
\end{align*}
\def\fig{swap-to-normal-form}
\begin{align*}
\zxc\vdash~~&
\eq{}
\eq{}\\&
\eq{\ref{lem:k1}\\\bo}
\eq{\ref{prop:tensor-pnf}}
\eq{\ref{prop:swap-pnf}}\input{./figures/\fig/\fig_05.tikz}
\end{align*}
Any green dot with arity larger than 3 can be decomposed as a 3-legged dots thanks to \soo, and any red dot is a green dot with Hadamard gates on its adjacent wires. Then, any diagram can be built from the states by simple topological transformations. E.g:
\def\fig{identity-to-normal-form}
\begin{align*}
\zxc\vdash\tikzstyle{every picture}=[baseline=0em]
\left(
\eq{}
\eq{}\right),
\def\fig{cup-to-normal-form}
\left(
\eq{}
\eq{}\right)
\end{align*}
\qed\end{proof}

\subsection{Completeness for the General ZX-Calculus}

\begin{proof}[Theorem \ref{thm:general-completeness}]
\textbf{$\Lambda$ is a representation of controlled states:} By definition, $\zx^{\textnormal A}$ proves the induction part. Let $x=\rho e^{i\theta}$. Then:
\def\fig{control-x-base-case-ket-0}
\begin{align*}
&
\eq{}
\eq{\bo\\\ref{lem:inverse}\\\soo\\\stt}
\eq{\stt\\\soo\\\bo}\\&
\eq{\ref{cor:arccos-2-power-n}}
\eq{\ref{lem:inverse}}\input{./figures/\fig/\fig_05.tikz}
\end{align*}
and:
\def\fig{control-x-base-case-ket-1}
\begin{align*}
\interp{}
&\eq[]{}\interp{}
\eq[]{\soo\\\bo\\\stt\\\ref{lem:k1}}\interp{}
\eq[]{\kt\\\ref{lem:multiplying-global-phases}\\\soo}\interp{}\\
&=e^{i(\theta-\beta)}\frac{2^n}{2}(1+e^{2i\beta})=2^n\cos(\beta)e^{i\theta}=\rho e^{i\theta}=x
\end{align*}
If $x=0$, these results are obvious.
\textbf{The conditions for applying Theorem \ref{them:completeness} are respected:}\\
$\bullet$ If either $x=0$ or $y=0$, the sum and product are obvious. Otherwise, let $x=\rho_1 e^{i\theta_1}$ and $y=\rho_2 e^{i\theta_2}$:
\def\fig{sum-of-control-scalars-proof-general}
\begin{align*}
&
\eq{}
\eq{\ref{lem:ctrl-power-2}}\\&
\eq{\ref{lem:k1}\\\ref{lem:not-triangle-is-symmetrical}\\\ref{lem:Calpha-through-W}}
\eq{\ref{cor:arccos-2-power-n}}\\&
\eq{\ref{lem:add-bis}}\input{./figures/\fig/\fig_05.tikz}
\eq{}\input{./figures/\fig/\fig_06.tikz}
\end{align*}
with 
\begin{align*}
\forall k\in\{1,2\},~\beta_k &= \arccos(\frac{\rho_k}{2^{n_k}})\qquad
\gamma_k = \arccos(\frac{1}{2^{n_k}})\qquad
\beta_2' = \arccos(\frac{\rho_2}{2^{n_1}})\\
\theta_3 &=\arg(\rho_1e^{i\theta_1}+\rho_2e^{i\theta_2})\\
\beta_3 &= \arccos(e^{i\theta_1-\theta_3}\cos{\beta_1}+e^{i\theta_2-\theta_3}\cos{\beta_2})= \arccos(\frac{\rho_1e^{i\theta_1}+\rho_2e^{i\theta_2}}{e^{i\theta_3}2^n})
\end{align*}
$\bullet$ 
\def\fig{prod-of-control-scalars-proof-general}
\begin{align*}
&
\eq{}
\eq{\soo\\\ref{lem:k1}}\\&
\eq{\ref{lem:prod-cos}}
\eq{}
\end{align*}
where
\begin{align*}
\forall k\in\{1,2\},~n_k&=\max\left(0,\left\lceil \log_2(\rho_k)\right\rceil\right),~\beta_k = \arccos(\frac{\rho_k}{2^{n_k}}),~\gamma_k = \arccos(\frac{1}{2^{n_k}})\\
\beta_3 &= \arccos(\frac{\rho}{2^{n_1+n_2}}),~
\gamma_3 = \arccos(\frac{1}{2^{n_1+n_2}})
\end{align*}
$\bullet$ 
\def\fig{condition-on-exp-i-alpha-proof-general}
\begin{align*}

\eq{}
\eq{}
\end{align*}
$\bullet$ and of course $\zx^{\textnormal A}\vdash\zx$.\\
Hence, we can use Theorem \ref{them:completeness}.
\qed\end{proof}

\subsection{Completeness for the \frag{4n}s}

\begin{proof}[Lemma \ref{lem:zx-not-complete}]
Let $p$ be an odd prime number and $\ell$ an integer $\geq 1$. The formula of the cyclotomic polynomial for a number with at most one odd prime factor gives:
$\phi_{8p^\ell}(x)=\sum\limits_{k=0}^{p-1}(-1)^kx^{4kp^{\ell-1}}$. Moreover, $(-1)^ke^{i\frac{\pi}{4p^\ell}\times 4kp^{\ell-1}}=e^{i\frac{p+1}{p}k\pi}$. After telescoping:
\def\fig{controlled-cyclo-8-p-to-l}
\[\eq{}\]
Since $p$ and $4$ are coprime, there exists $k$ such that $kp\frac{\pi}{4}=\frac{\pi}{4}$. Let us then consider the interpretation $[.]_{kp}$ which multiplies all the angles by $kp$: $D_1\otimes D_2\mapsto [D_1]_{kp}\otimes [D_2]_{kp}$, $D_1\circ D_2\mapsto [D_1]_{kp}\circ [D_2]_{kp}$, $R_Z^{(n,m)}(\alpha)\mapsto R_Z^{(n,m)}({kp}\alpha)$, $R_X^{(n,m)}(\alpha)\mapsto R_X^{(n,m)}({kp}\alpha)$, $Id$ otherwise. It is routine to show that the rules of \zx hold under this interpretation, but:
\[~~\mapsto~~~~\neq ~~~~\mapsfrom~~\]
\qed\end{proof}

\begin{proof}[Lemma \ref{lem:sum-prod-polynomials}]
First, if $x,y\in\mathbb{N}$:
\def\fig{sum-of-control-x-aux}
\begin{align*}

\eq[]{\ref{lem:gn-distrib-over-W}\\\ref{lem:triangle-through-W}}
\eq[]{\ref{lem:gn-distrib-over-W}\\\ref{lem:triangle-through-W}}
\eq[]{\ref{lem:gn-distrib-over-W}}
\end{align*}
If $r=s$:
\def\fig{sum-of-control-x-aux-0}
\begin{align*}

\eq{\stt}
\end{align*}
Otherwise, if $r\neq s$ and $x\geq y$:
\def\fig{sum-of-control-x-aux-1}
\begin{align*}

\eq{\stt\\\soo}
\eq{\ref{lem:inverse-of-triangle}}
\end{align*}
The case $r\neq s$ and $x\leq y$ is similar.
In the end:
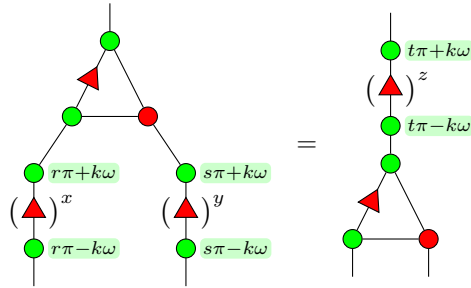
\def\fig{sum-of-control-x-aux}
\begin{align*}

\eq{}
\end{align*}
with $(-1)^tz=(-1)^rx+(-1)^sy$. The result for the sum immediately follows by induction (if $0$ is involved, the result is obvious).
For the product, first, if $P(X)=P'(X)+(-1)^baX^k$:
\def\fig{gamma_alpha-times-gn-alpha}
\begin{align*}
\label{eq:gamma-alpha-times-gn}

\eq{}\!
\eq{\soo}\!
\eq{}\!\!
\eq{}
\end{align*}
and \def\fig{gamma_alpha-times-gn-alpha-0}$\eq{}\eq{}\eq{}$.
Then, if $Q$ is non-null:
\def\fig{prod-of-control-polynomials-proof}
\begin{align*}                                                                

\eq{}
\eq{(\ref{lem:polynomial-through-triangle})}
\eq{\ref{lem:gn-distrib-over-W}}\\
\eq{}
\eq{}\input{./figures/\fig/\fig_05.tikz}
\eq{}\input{./figures/\fig/\fig_06.tikz}
\eq{}\input{./figures/\fig/\fig_07.tikz}
\end{align*}
and if $Q=0$, the result is obvious.
\qed\end{proof}

\begin{proof}[Proposition \ref{lem:cyclo-to-0}]
First of all, we can easily derive for any $N$:
\[\zx_{\frac{\pi}{4n}}\vdash~~
\input{./figures/poly-x-to-n-minus-1.tikz}
~~\implies~~\zx_{\frac{\pi}{4n}}\vdash~~
\input{./figures/cyclo-1.tikz}
\]
Now, assume $p$ is prime. Then, $\phi_1(X)\phi_d(X)=\prod\limits_{d|p}\phi_d(X)=X^p-1$. Since sums and products of control polynomials are derivable in ZX, it means:
\def\fig{cyclo-prime}
\begin{align*}
\zx_{\frac{\pi}{4n}}\vdash~~\eq{}
\equi{}\eq{}\\
\equi{\ref{lem:polynomial-times-not-polynomial}}\eq{}\input{./figures/\fig/\fig_05.tikz}
~~\implies~~\input{./figures/\fig/\fig_06.tikz}\eq{}\input{./figures/\fig/\fig_07.tikz}\\
\equi{\ref{lem:pi-green-state-on-upside-down-triangle}}\input{./figures/\fig/\fig_08.tikz}\eq{}\input{./figures/\fig/\fig_09.tikz}
\equi{\kt\\\bo\\\ref{lem:pi-red-state-on-triangle}}\input{./figures/\fig/\fig_10.tikz}\eq{}\input{./figures/\fig/\fig_11.tikz}\\
\equi{\ref{lem:multiplying-global-phases}\\\ref{lem:inverse}}\input{./figures/\fig/\fig_12.tikz}\eq{}\input{./figures/\fig/\fig_13.tikz}
~~\implies~~\zx^{\textnormal{(cancel)}}_{\frac{\pi}{4n}}\vdash~~\input{./figures/\fig/\fig_14.tikz}\eq{}\input{./figures/\fig/\fig_15.tikz}
\end{align*}
Now, if $p$ is still prime, the case of $p^k$ is handled with the equation $\phi_{p^k}(X)=\phi_p(X^{p^{k-1}})$ which translates as:
\def\fig{cyclo-power-prime}
\begin{align*}
\zx_{\frac{\pi}{4n}}\vdash~~\eq{}
~~\implies~~\zx^{\textnormal{(cancel)}}_{\frac{\pi}{4n}}\vdash~~\eq{}\eq{}
\end{align*}
Finally, in the general case, let $8n=\prod_i p_i^{k_i}$ with all $p_i$ primes. Then, $\phi_{8n}(X)=\gcd\limits_i\left(\phi_{p_i^{k_i}}(X^{p_i^{k_i-1}})\right)$. By B\'ezout's identity, $\phi_{8n}(X)=\sum\limits_i Q_i(X)\phi_{p_i^{k_i}}(X^{p_i^{k_i-1}})$ where the $Q_i$ are some unitary polynomials. This translates as:
\def\fig{cyclo-prime-decomp}
\begin{align*}
\zx_{\frac{\pi}{4n}}\vdash~~&\eq{}\\\\
\implies\zx^{\textnormal{(cancel)}}_{\frac{\pi}{4n}}\vdash~~&
\eq{}\!\!\\\\&
\eq[]{}
\eq[]{}\input{./figures/\fig/\fig_05.tikz}\\\\&
\eq{}\cdots\eq{}\input{./figures/\fig/\fig_06.tikz}
\eq{}\input{./figures/\fig/\fig_07.tikz}
\end{align*}
\qed\end{proof}

\begin{proof}[Theorem \ref{thm:pi_4n-completeness}]
The product is obvious when we have Lemmas \ref{lem:sum-prod-polynomials} and \ref{lem:cyclo-to-0}. For the sum, let $x=\frac{1}{2^p}P(e^{i\frac{\pi}{4n}})$, $y=\frac{1}{2^q}Q(e^{i\frac{\pi}{4n}})$. W.l.o.g., assume $p\leq q$. Then:
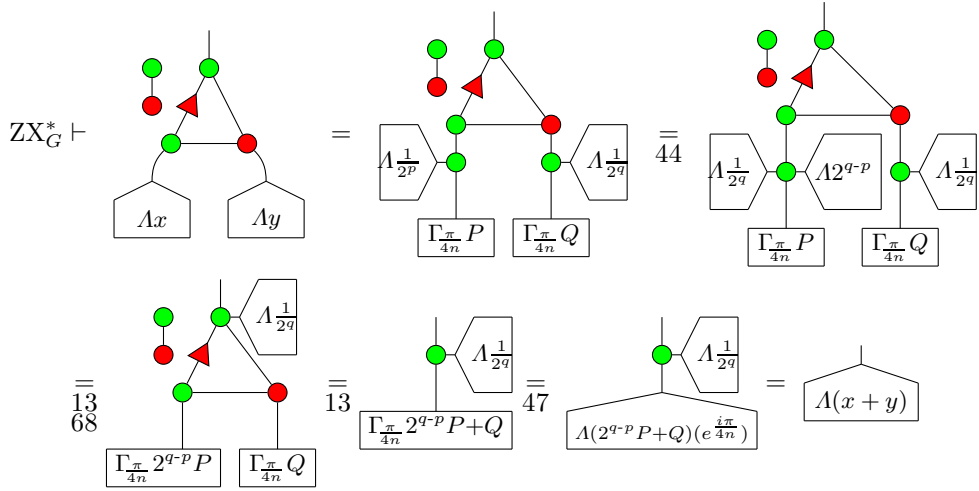
\def\fig{sum-of-control-scalars-proof-pi_4n}
\begin{align*}
\zxc&\vdash~~
\eq{}
\eq{\ref{lem:2Id-is-C2-times-anti-C2}}\\
&\eq[]{\ref{lem:sum-prod-polynomials}\\\ref{lem:C2-distributed-through-W-2}}
\eq[]{\ref{lem:sum-prod-polynomials}}
\eq[]{\ref{lem:2-through-triangle}}~\input{./figures/\fig/\fig_05.tikz}
\eq[]{}~\input{./figures/\fig/\fig_06.tikz}
\end{align*}
The ante-penultimate diagram may not directly be in normal form, for there may be $S$ such that $2^{q-p}P+Q=2S$, but this is dealt with with Lemma \ref{lem:2-through-triangle}.
\qed\end{proof}

\subsection{Completeness for the \frag{2^n}}

\begin{proof}[Lemma \ref{lem:pi_2^n-inverse}]
If $k\in\{-2^n+1,\cdots, 2^{n+1}-1\}$, then there exist $0\leq m<n$ and $p\in\mathbb{Z}$ such that $k=2^m(2p-1)$i.e.~$\frac{k\pi}{2^n}=\frac{2p-1}{2^{n-m}}\pi$ where $2^{n-m}\geq2$. Then:
\def\fig{gn-k-pi_2-to-n-inverse}
\begin{align*}

\eq{}
\eq{} \cdots\\
\eq{}
\eq{}
\eq{}
\end{align*}
\qed\end{proof}

\begin{proof}[Corollaries \ref{cor:qpi-completeness} and \ref{cor:dpi-completeness}]
Let $G$ be a subgroup of $\mathbb{Q}\pi$, and $D_1$ and $D_2$ be two diagrams of the fragment $G$, such that $\interp{D_1}=\interp{D_2}$. If $G$ is finite, Theorem \ref{thm:pi_4n-completeness} directly gives the result. Otherwise, there exists $n\in\mathbb{N}$ such that $\frac{\pi}{4n}\in G$ and both diagrams are in the \frag{4n} of the ZX-Calculus. By completeness (Theorem \ref{thm:pi_4n-completeness}): $\zx_{\frac{\pi}{4n}}^{\textnormal{cancel}}\vdash D_1=D_2$, and since $\zx_{G}^{\textnormal{cancel}}\vdash\zx_{\frac{\pi}{4n}}^{\textnormal{cancel}}$, $\zx_{G}^{\textnormal{cancel}}\vdash D_1=D_2$.

The proof for the $\mathbb{D}\pi$ fragment is similar, except we use the completeness of \frag{2^n} (Theorem \ref{thm:pi_2-to-n-completeness}), and the set of axioms \zx.
\qed\end{proof}

\end{document}